\documentclass[a4paper,UKenglish,cleveref, autoref, thm-restate]{lipics-v2021}

\title{Encoding impredicative hierarchy of type universes with variables} 

\author{Yoan Géran}{Mines Paris - PSL, Centre de Recherche en Informatique \and Université Paris-Saclay, Laboratoire Méthodes Formelles, ENS Paris-Saclay}{yoan.geran@minesparis.psl.eu}{}{}

\authorrunning{Y. Géran}
\Copyright{Yoan Géran}



\category{} 

\relatedversion{} 

\supplement{An implementation is available at \url{https://gitlab.crans.org/geran/dedukti-level-implementation}}

\acknowledgements{I want to thank my PhD advisors Olivier Hermant and Gilles Dowek for
the helpful discussions and comments.}

\nolinenumbers 

\EventEditors{John Q. Open and Joan R. Access}
\EventNoEds{2}
\EventLongTitle{42nd Conference on Very Important Topics (CVIT 2016)}
\EventShortTitle{CVIT 2016}
\EventAcronym{CVIT}
\EventYear{2016}
\EventDate{December 24--27, 2016}
\EventLocation{Little Whinging, United Kingdom}
\EventLogo{}
\SeriesVolume{42}
\ArticleNo{23}

\usepackage{mathtools, amssymb, stmaryrd}
\usepackage{fancyvrb}
\usepackage{correct_math}
\usepackage{new_operator}
\usepackage[letter]{french_set}

\usepackage{tikz}

\usepackage{csquotes}

\newoperator{\prodrule}{imax}
\renewoperator{\R}{imax}

\newoperator{\Var}{\mathcal{V}}

\newcommand*{\lambdapimod}{\lambda\Pi/\equiv}
\newcommand*{\X}{\mathcal{X}}

\newcommand*{\Term}{\mathcal{L}}
\newcommand{\SubTerm}{\Term_{\text{s}}}
\newcommand{\Sub}{\text{s}}

\newcommand{\NFTerm}{\Term_{r}}

\newcommand{\NF}{\text{r}}

\newcommand{\eqTerm}{=_{\Term}}
\newcommand{\leqTerm}{\leqslant_{\Term}}

\newoperator{\A}{\mathcal{A}}
\newoperator{\Ap}{\mathcal{A}'}
\newoperator{\B}{\mathcal{B}}

\newoperator{\nf}{nf}

\newoperator{\simplify}{simplify}

\newoperator{\simpl}{simpl}

\let\rewrites\hookrightarrow

\newcommand*{\enlarge}{\vphantom{\mathord{\big(}}}
\newcommand*{\interprete}[1]{
	\left\llbracket #1 \right\rrbracket
}

\newcommand*{\valuation}[2]{
	\interprete{#1}_{#2}
}

\newcommand{\I}{\text{\textnormal{\texttt{I}}}}

\newcommand{\Nat}{\text{\textnormal{\texttt{N}}}}
\newoperator{\s}{\text{\textnormal{\texttt{s}}}}
\newcommand{\zero}{\text{\textnormal{\texttt{0}}}}
\newcommand{\plus}{\mathbin{\text{\textnormal{\texttt{+}}}}}
\newoperator{\maxN}{\text{\textnormal{\texttt{max}}}_{\Nat}}

\newcommand{\Bool}{\text{\textnormal{\texttt{B}}}}
\newcommand{\true}{\mathalpha{\text{\textnormal{\texttt{true}}}}}
\newcommand{\false}{\text{\textnormal{\texttt{false}}}}
\newoperator{\ands}{\text{\textnormal{\texttt{and}}}}
\newoperator{\ors}{\text{\textnormal{\texttt{or}}}}
\newoperator{\nots}{\text{\textnormal{\texttt{not}}}}

\newcommand*{\Set}[1]{\text{\textnormal{\texttt{S[$#1$]}}}}
\newcommand*{\cons}[1]{\mathtt{::}_{#1}}
\newcommand*{\nil}[1]{\text{\textnormal{\texttt{\{\}}}}_{#1}}
\newcommand*{\add}[1]{\mathbin{\text{\textnormal{\texttt{<<}}}}_{#1}}
\newcommand*{\union}[1]{\mathbin{\text{\textnormal{\texttt{++}}}}_{#1}}
\newcommand*{\included}[1]{\mathbin{\text{\textnormal{\texttt{in}}}}_{#1}}
\newcommand*{\sub}[1]{\mathbin{\text{\textnormal{\texttt{sub}}}}_{#1}}
\newcommand*{\del}[1]{\mathbin{\text{\textnormal{\texttt{\textbackslash}}}}_{#1}}

\newcommand*{\setS}[2]{\text{\textnormal{\texttt{\{$#2$\}}}}_{#1}}

\newcommand{\NatSet}{\Set{\Nat}}

\newoperator{\iteop}{\text{\textnormal{\texttt{ite}}}}

\newcommand*{\eqs}[1]{\mathbin{\text{\textnormal{\texttt{=}}}_{#1}}}
\newcommand*{\leqs}[1]{\mathbin{\text{\textnormal{\texttt{<=}}}_{#1}}}
\newcommand*{\lqs}[1]{\mathbin{\text{\textnormal{\texttt{<}}}_{#1}}}

\newcommand{\Level}{\text{\textnormal{\texttt{L}}}}
\newcommand{\zeroL}{\text{\textnormal{\texttt{0}}}_{\Level}}
\newoperator{\sL}{\text{\textnormal{\texttt{s}}}_{\Level}}
\newoperator{\maxL}{\text{\textnormal{\texttt{max}}}_{\Level}}
\newoperator{\varL}{\text{\textnormal{\texttt{var}}}_{\Level}}
\newoperator{\ruleL}{\text{\textnormal{\texttt{imax}}}_{\Level}}

\newcommand{\SLevel}{\text{\textnormal{\texttt{L}}}_{\text{\textnormal{\texttt{S}}}}}
\newoperator{\sSL}{\text{\textnormal{\texttt{s}}}_{\SLevel}}

\newoperator{\ruleSL}{\text{\textnormal{\texttt{imax}}}_{\SLevel}}

\newoperator{\As}{\text{\textnormal{\texttt{a}}}}
\newoperator{\Bs}{\text{\textnormal{\texttt{b}}}}
\newoperator{\addable}{\text{\textnormal{\texttt{add?}}}}

\newcommand{\SLSet}{\Set{\SLevel}}

\newoperator{\maxhelper}{\text{\textnormal{\texttt{maxhelper}}}}
\newoperator{\rulehelper}{\text{\textnormal{\texttt{imax\_aux}}}}

\newoperator{\maxS}{\text{\textnormal{\texttt{max}}}_{\SLevel}}

\newoperator{\Prf}{\text{\textnormal{\texttt{Prf}}}}

\newcommand*{\order}[1]{\mathbin{\text{\textnormal{\texttt{leq}}}}_{#1}}
\newcommand*{\strictOrder}[1]{\mathbin{\text{\textnormal{\texttt{lq}}}}_{#1}}

\newcommand{\voidP}{\text{\texttt{\_}}}

\newoperator{\evalS}{\text{\textnormal{\texttt{eval}}}_{\SLevel}}
\newoperator{\evalL}{\text{\textnormal{\texttt{eval}}}_{\Level}}

\newcommand*{\ifs}[1][]{\mathop{\text{\textnormal{\texttt{if}}}}}
\newcommand*{\thens}[1][]{\mathop{\text{\textnormal{\texttt{then}}}}}
\newcommand*{\elses}[1][]{\mathop{\text{\textnormal{\texttt{else}}}}}

\newoperator{\setPart}{\mathcal{P}}

\newoperator{\succSub}{s_{\Sub}}

\newoperator{\maxSub}{max_{\Sub}}

\newcommand*{\translation}[1]{\left|#1\right|}

\newcommand*{\subst}[2]{\left[#1\right]#2}

\newcommand{\Dedukti}{\textsc{Dedukti}}

\newcommand{\rewriteDef}{\longrightarrow}

\newcommand{\tese}{\vdash}

\newcommand{\Axiom}{\mathcal{A}}
\newcommand{\Rule}{\mathcal{R}}
\newcommand{\Sort}{\mathcal{S}}

\usepackage{ebproof}

\newoperator{\repr}{repr}

\newoperator{\Type}{Type}
\newoperator{\Kind}{Kind}
\newoperator{\Prop}{Prop}

\newcommand{\sType}{\text{\texttt{Type}}}

\newoperator{\UType}{\text{\texttt{U}}}
\newoperator{\UTerm}{\text{\texttt{u}}}
\newoperator{\El}{\text{\texttt{El}}}

\newcommand*{\Agda}{\textsc{Agda}}
\newcommand*{\Coq}{\textsc{Coq}}
\newcommand*{\Lean}{\textsc{Lean}}
\newcommand*{\Matita}{\textsc{Matita}}

\newcommand*{\newrule}[1]{\left(#1\right)}

\newoperator{\imax}{imax}

\newoperator{\simax}{\text{\textnormal{\texttt{imax}}}}
\newoperator{\smax}{\text{\textnormal{\texttt{max}}}}

\newcommand*{\app}[2]{\mathalpha{#1}\;\mathalpha{#2}}

\newcommand{\ICC}{\text{CC}^{\infty}}

\newcommand\TTTT{%
	\textsf{T\kern-0.15em\raisebox{-0.55ex}T\kern-0.15emT\kern-0.15em\raisebox{-0.55ex}2}%
}

\newcommand{\Lambdapi}{\textsc{Lambdapi}}
\newcommand{\Kontroli}{\textsc{Kontroli}}

\newcommand{\HOL}{\textsc{HOL-Light}}

\begin{document}

\maketitle

\begin{abstract}
Logical frameworks can be used to translate proofs from a proof system to another one. 
For this purpose, we should be able to encode the theory of the proof system in the logical framework. 
The Lambda Pi calculus modulo theory is one of these logical frameworks. Powerful theories such as 
pure type systems with an infinite hierarchy of universes have been encoded, leading to partial encodings 
of proof systems such as Coq, Matita or Agda. In order to fully represent systems such as Coq and Lean, we 
introduce a representation of an infinite universe hierarchy with an impredicative universe and universe
variables where universe equivalence is equality, and implement it as a terminating and confluent rewrite system.
\end{abstract}

\section{Introduction}

The formalization of mathematical theorems and the verification of softwares
are done in several tools, and many logical systems and theories are developed
as the research on proof-checking makes progress. Interoperability is then
a big challenge which aims to avoid the redevelopment of the same proof in
each system. Instead of developing translators from each system to another
ones, \emph{logical frameworks} proposes to define theories in a common language, which makes makes translation easier. Thus, the logical framework should be
expressive enough and work should be done to define the wanted theories in 
the framework. 

In this paper, our goal is to show how to define the universe levels of 
the theory of the \Coq{} proof system in one of these framework, the $\lambda\Pi$-calculus 
modulo rewriting. Since a lot of theories are expressed as extensions of \emph{Pure Type Systems}, 
the first part of this introduction will define them. Then, we will present the 
$\lambda\Pi$-calculus modulo rewriting and the type system behind Coq, in particular the
universe levels.

\paragraph*{Pure Type Systems}

A lot of theories are based on extensions on Church's simply-typed $\lambda$-calculus (STLC). In \cite{BarendregtCube}, Barendregt introduced the $\lambda$-cube which classifies type systems depending on the possibility to quantify on types or terms to build new types or new terms. It captures systems such as System F (with type polymorphism), $\lambda\underline{\omega}$ (with type operators), $\lambda\Pi$ (with dependent types), or the Calculus of Constructions (CC) which allows all these quantifications.

More generally, these constructions can be extended, leading to more powerful systems called \emph{Pure Type Systems} \cite{BarendregtPTS, berardiPTS}.

\begin{definition}
	A Pure Type System (PTS) is defined by a set of sorts $\Sort$ (that we will also call universes), a set of axioms $\Axiom \subseteq \Sort^2$
	and a set of rules $\Rule \subseteq \Sort^3$. 
\end{definition}
$\Axiom$ describes the sorts typing ($s_1$ has the type $s_2$ when $(s_1, s_2) \in \Axiom$),
and $\Rule$ describes the possible quantifications and their typing rules. 
The terms are the following, where $s \in \Sort$ and $x$ is an element of a countable
set of variables $\X$.
\[
t \coloneqq s \mid x \mid (x \colon t) \to t \mid 
(\lambda x\colon t \cdot u) \mid tt
\]
and the typing rules are given in \cref{fig-pts}.
\begin{center}
	(\textsc{Empty})
	\begin{prooftree}
		\infer0{[\phantom{\cdot}] \text{ WF}}
	\end{prooftree}	
	\hfill 	
	(\textsc{Decl}) \begin{prooftree}
		\hypo{\Gamma \tese A \colon s}
		\hypo{x \not\in \Gamma}
		\infer2{\Gamma, x \colon A \text{ WF}}
	\end{prooftree}	
	\hfill 	
	(\textsc{Var}) \begin{prooftree}
		\hypo{\Gamma \text{ WF}}
		\hypo{(x \colon A) \in \Gamma}
		\infer2{\Gamma \tese x \colon A}
	\end{prooftree}	
	
	\vspace{1em}
	
	(\textsc{Sort}) \begin{prooftree}
		\infer0[$(s_1, s_2) \in \Axiom$]
		{\tese s_1 \colon s_2}
	\end{prooftree}	
	\hfill 
	(\textsc{Prod}) \begin{prooftree}
		\hypo{\Gamma \tese A \colon s_1}
		\hypo{\Gamma, x \colon A \tese A \colon s_2}
		\infer2[$(s_1, s_2, s_3) \in \Rule$]
		{\Gamma \tese \Pi x \colon A \cdot B \colon s_3}
	\end{prooftree}	
	
	\vspace{1em}
	
	(\textsc{App}) \begin{prooftree}
		\hypo{\Gamma \tese t \colon \Pi x \colon A \cdot B}
		\hypo{\Gamma, \tese u \colon A}
		\infer2{\Gamma \tese \app{t}{u} \colon B}
	\end{prooftree}	
	\hfill	
	(\textsc{Abs}) \begin{prooftree}
		\hypo{\Gamma, x \colon A \tese t \colon B}
		\hypo{\Gamma \tese \Pi x \colon A \cdot B \tese s}
		\infer2{\Gamma \tese \lambda x \cdot t \colon B}
	\end{prooftree}	
	
	\vspace{1em}
	
	(\textsc{Conv}) \begin{prooftree}
		\hypo{\Gamma \tese A \colon s}
		\hypo{\Gamma, \tese t \colon A}
		\hypo{A \equiv_{\beta} B}
		\infer3[$s \in \Sort$]{\Gamma \tese t \colon B}
	\end{prooftree}	
	\captionof{figure}{Typing rules}
	\label{fig-pts}
\end{center}

\begin{definition}[Functional and full PTS]
	A PTS is said functional if $\Axiom$ and $\Rule$
	are functional relations from $\Sort$ and $\Sort \times \Sort$ to $\Sort$, that is to say
	$(s_1, s_2) \in \Axiom \land (s_1, s_3) \in \Axiom \implies s_2 = s_3$ and
	$(s_1, s_2, s_3) \in \Rule \land (s_1, s_2, s_4) \in \Rule 
	\implies s_3 = s_4$.
	
	A PTS is called full if $\Axiom$ and $\Rule$ are total
	functions from $\Sort$ and $\Sort \times \Sort$ to $\Sort$.
\end{definition}

\paragraph*{The $\lambda\Pi$-calculus modulo rewriting}


$\lambda\Pi$, the extension of STLC with dependent types, is the language of the 
\emph{Edinburgh Logical Framework} (ELF) \cite{ELF}. However, computation plays 
an essential role in type theories, then in modern proof assistant, and 
$\lambda\Pi$ is not well-suited for this. To address this point, the $\lambda\Pi$-calculus 
modulo rewriting ($\lambdapimod$) \cite{lambdapimoduniversal} extends $\lambda\Pi$ by allowing 
user-defined higher-order rewrite rules \cite{jouannaudRewrite, terese} 
that can be used to define functions but also types. Types are then identified
modulo $\beta$ and these rewrite rules. 

 
In order to have good properties such as the decidability of the type-checking, the rewrite rules introduced 
should preserves typing and form a confluent and strongly normalizing rewrite system, which adds some 
restrictions and requires more efforts to show that these properties are respected.

\begin{remark}
In the rest of this article, we will use the syntax $u \rewriteDef v$ (where $u$ may contains free variables used for matching) to define a rewrite 
rule and the syntax $u \rewrites v$ to indicate that the term $u$ rewrites to the term $v$.
\end{remark}

The $\lambdapimod$ can express CC and its subtheories \cite{theoryU}, and, in \cite{DowekCousineau}, 
Cousineau and Dowek show how to embed functional PTS (with a possibly infinite number of symbols and rules):
\begin{enumerate}
	\item for each sort $s$, symbols $\UType_{\mathtt{s}} \colon \sType$ and $\El_{\mathtt{s}} \colon \UType_{\mathtt{s}} \to \Type$,
	\item for each axiom $(s_1, s_2)$, a symbol $\UTerm_{\mathtt{s_1}} \colon \UType_{\mathtt{s_2}}$ and a rewrite
	      rule $\El_{\mathtt{s_2}}(s_1) \rewriteDef \UType_{\mathtt{s_1}}$,
	\item for each rule $(s_1, s_2, s_3)$, a symbol 
	      $\pi_{\mathtt{(s_1, s_2)}} \colon (x \colon \UType_{\mathtt{s_1}}) \to (\El_{\mathtt{s_1}}(x) \to \UType_{\mathtt{s_2}}) \to \UType_{\mathtt{s_3}}$
	      and a rewrite rule 
	      $\El_{\mathtt{s_3}}(\pi_{\mathtt{(s_1, s_2)}}(A, B)) \rewriteDef 
	      (x \colon \El_{\mathtt{s_1}}(A)) \to \El_{\mathtt{s_2}}(B)$.
\end{enumerate} 
$\UTerm_{\mathtt{s}}$ corresponds to the sort $s$ as a term of type $s'$, 
$\UType_{\mathtt{s}}$ to the sort $s$ as a type, and $\El_{\mathtt{s}}$ associates a sort (as term) of type $s$ to its corresponding type, hence
the rewrite rule added for each axiom. And $\pi_{\mathtt{(s_1, s_2)}} A B$ is the term corresponding to the types  of the function from $A$ to $B(A)$, hence the rewrite rule added to obtain the type associated to this term.

Several systems have been encoded in $\lambdapimod$: \HOL{} \cite{Thire, theseAssaf}, 
\Agda{} \cite{genestier_universes}, \Matita{} \cite{theseAssaf}, but also parts of \Coq, on which we will come back to later. Besides, 
since there exists multiple implementations of the $\lambdapimod$ such as \Dedukti{} \cite{dedukti}, 
\Lambdapi{} \cite{lambdapi}, or \Kontroli{} \cite{kontroli}, these embeddings have been effectively
implemented leading to translations from the proofs systems to these implementations, but also to 
translations from these implementations of $\lambdapimod$ to proofs assistants \cite{Thire, theseThire}.

\paragraph*{\Coq's type system}

The theory of \Coq{} is based on CC extended with an infinite hierarchy of universes and an
impredicative universe $\Prop$. It corresponds to a slightly different version of the following 
PTS (where $\Prop$ is denoted as $\Type_0$).

\begin{definition}[Impredicative max]
	We define $\imax \colon \N \to \N \to \N$ by $\imax[i, 0] = 0$ and
	$\imax[i, j + 1] = \max[i, j + 1]$.
\end{definition}

\begin{definition}[$\ICC$]
	$\ICC$ is the full PTS defined with an infinite sequence of sorts $\Type_i$ indexed on $\N$, the 
	axioms $\newrule{\Type_i, \Type_{i + 1}}$, and the rules $\newrule{\Type_i, \Type_j, \Type_{\imax[i, j]}}$.
\end{definition}

\begin{remark}
One can also define a predicative PTS where the products from $\Type_i$ to $\Type_j$ are 
elements of $\Type_{\max[i, j]}$ instead of $\Type_{\imax[i, j]}$. This latter is the building on \Agda{}
proof system while $\ICC$ is the one of \Coq{} but also of \Lean{} or \Matita.
\end{remark}

In both cases, the sorts are characterized by \emph{levels} indexed on $\N$;
the functions $\Rule$ and $\Axiom$ can be defined in the $\lambdapimod$, and
we can adapt the general embedding of Cousineau and Dowek to a finite
embedding. For that, we define the natural numbers $\Nat$ with the successor function $\s$, $\smax \colon \Nat \to \Nat \to \Nat$, $\simax \colon \Nat \to \Nat \to \Nat$, and
\begin{itemize}
	\item symbols $\UType \colon \Nat \to \sType$ and 
	      $\El \colon (i \colon \Nat) \to \UType(i) \to \Type$,
	\item a symbol $\UTerm \colon (i \colon \Nat) \to \UType (\s(i))$ and a rewrite
	rule $\app{\app{\El}{\voidP}}{i} \rewriteDef \app{\UType}{i}$,
	\item a symbol $\pi \colon (i \colon \Nat) \to (j \colon \Nat) \to (A \colon 
	\app{\UType}{i}) \to (\app{\app{\El}{i}}{A} \to \app{\UType}{j}) \colon \app{\UType}{(\app{\app{\simax}{i}}{j})}$
	and a rewrite rule 
	$\app{\app{\El}{\voidP}}{(\app{\app{\app{\app{\pi}{i}}{j}}{A}}{B})} \rewriteDef 
	(x \colon \app{\app{\El}{i}}{A}) \to \app{\app{\El}{j}}{(\app{B}{x})}$.
\end{itemize}

\Coq{} extend $\ICC$ with other features. Some of them have been encoded, 
leading to a partial translator from \Coq{} to \Dedukti{} \cite{Burel, theseFerey}, and to 
the sharing of developements of the GeoCoq library \cite{geocoq}, a formalization of geometry,
to other proof assistants \cite{sttfageocoq}. Extensions such as inductive types \cite{PaulinCoquand,Paulin}
and cumulativity \cite{luocumul} have been widely covered: Burel and Boespflug in
\cite{Burel}, then Férey in his thesis \cite{theseFerey} propose embeddings of inductive constructions 
and cumulativity have been studied by Assaf \cite{assaf_explicit_subtyping, theseAssaf},
Férey \cite{theseFerey} or Thiré \cite{theseThire}. 

In this paper, we are interested in another feature, the level variables, which permits to extend $\ICC$ with floating universes \cite{HuetType} or 
with universe polymorphism \cite{SozeauTabareau,HarperPollack,Courant}.  

\paragraph*{Level variables}

We extend the syntax of the levels with variables.

\begin{definition}[Levels]
	A level is a term of the grammar 
	\[
		t \coloneqq 0 \mid s(t) \mid \max[t, t] \mid \imax[t, t]
	\]
	where $x$ is an element of a countable set of variables $\X$.
	We denote by $\Term$ the set of the levels, and we say that $t$ is a concrete levels if $t$
	does not contain any variable.
\end{definition}

\begin{definition}[Valuation]
	A valuation is a function $\sigma \colon \X \to \N$.
\end{definition}

\begin{definition}
	Let $\sigma \colon \X \to \N$ be a a valuation. We define 
	inductively the value of a level $t$ over $\sigma$,
	denoted as $\valuation{t}{\sigma}$
	with
	\begin{gather*}
		\valuation{0}{\sigma} = 0 \qquad
		\valuation{s(t)}{\sigma} = s(\valuation{t}{\sigma}) \qquad
		\valuation{x}{\sigma} =  \sigma(x)\\
		\valuation{\max[t_1, t_2]}{\sigma} = 
		\max[\enlarge\valuation{t_1}{\sigma}, \valuation{t_2}{\sigma}]\qquad
		\valuation{\prodrule[t_1, t_2]}{\sigma} = 
		\prodrule[\enlarge\valuation{t_1}{\sigma}, \valuation{t_2}{\sigma}]
	\end{gather*}
\end{definition}

We use the same symbol $s$, $\max$, and $\imax$ for the syntax of the levels and the functions of the
natural numbers. However, levels are abstract terms and are interpreted through valuations.  
Besides, the concrete levels can clearly be identified as the natural numbers which justifies 
the use of the same symbol and permits to see the interpretation as a function that
\emph{concretizes} a level, turning it into a concrete level. 

\begin{definition}[Level comparison]
	Let $t_1, t_2 \in \Term$. We say that $t_1 \leqTerm t_2$
	if for all valuations $\sigma$,
	$\valuation{t_1}{\sigma} \leqslant \valuation{t_2}{\sigma}$.
	In the same way, we say that $t_1 \eqTerm t_2$ if for all valuations
	$\sigma$, $\valuation{t_1}{\sigma} = \valuation{t_2}{\sigma}$.
	Hence $t_1 \eqTerm t_2$ if and only if $t_1 \leqTerm t_2$ and $t_2 \leqTerm t_1$.
\end{definition}

With level variables, the equivalence is no more the syntactic equality and the above embedding does not reflect it anymore: $\smax[x, y]$ and $\smax[y, x]$ are not convertible and adding rules for that would lead to a non-terminating system. In the same way, commutativity and
equivalences such as $\max[x, x] \eqTerm x$ or $\max[s(x), x)] \eqTerm s(x)$ are hard to express, and the impredicatity introduces other:$\prodrule[x, x] \eqTerm x$, $\max[\R[x, y], x] \eqTerm \max[x, y]$, etc.

And yet, a correct embedding of the levels should reflect these equivalences. For instance, a term of the 
universe $\Type_x$ is also a term of the universe $\Type_{\imax[x, x]}$, and then we should be able to 
identify the universes such as $\Type_x$ and $\Type_{\imax[x, x]}$. This paper presents a new embedding 
that faithfully represents levels with variables.

\paragraph*{Related work}

Some solutions have been studied in the predicative case. The big issue is the associativiy and commutativity of the $\max$ symbol. In \cite{genestier_universes}, Genestier solved this problem to encode \Agda's universe polymorphism. For that, he used rewriting modulo associativity and commutativity (AC). The idea,
also mentionned in a draft of Voevodsky \cite{voevodsky}, is to represent each level as 
$\max[n, n_1 + x_1, \ldots, n_k + x_k]$ where $n \geqslant \max[n_1, \ldots, n_k]$. Besides, if there exists $i \neq j$ such that $x_j = x_i$, we simplify the term and keep only
$\max[n_i, n_j] + x_i$. Then, we obtain a minimal representation of terms of the $\max$-successor algebra.

Blanqui gives another presentation of this algebra in \cite{blanqui_universes}, with an encoding without matching modulo AC. However, this solution requires to keep the level in 
some AC canonical form, and can then require the modification of the $\lambdapimod$ type-checker.

The $\imax$-successor algebra is less studied and we do not know easy ways to reflect its equalities.
A confluent encoding is proposed in \cite{assaf_universes}, but it does not fully reflect the equalities; for instance, the levels $\max[\prodrule[x, y], x]$ and $\max[x, y]$ are not convertible. 
Besides, Férey designed a non-confluent encoding of universe polymorphism in \cite{theseFerey}.

\paragraph*{Contribution and outline}

We introduce a new representation for the levels, using the idea presented above in the predicative case: find a set of subterms such that any level can be expressed as a maximum of subterms. They should be easily comparable to simplify $\max[u, v]$ into $u$ if $u \leqTerm v$ and obtain minimal
representations, and they should ensure the uniqueness property: 
\begin{gather*}
	\max[u_1, \ldots, u_n] \eqTerm \max[v_1, \ldots, v_m] \iff
	\set{u_1, \ldots, u_n} = \set{v_1, \ldots, v_m}.
\end{gather*}
Intuitively, the subterms should be very basic and simple: a subterm $u$ must not be equivalent to a maximum of other subterms. With this representation, we obtain a deep understanding of the $\imax$-successor algebra,
and an easy procedure-decision for the level inequality problem.

In the \cref{sec-level-transformation}, we study the semantic of the $\prodrule$ operator and establish a suitable set of subterms. Then, in the \cref{sec-uniqueness}, we introduce the minimal representation and shows that equivalent terms have the same minimal representation. And the \cref{sec-implementation} is dedicated to the implementation of this representation into the $\lambdapimod$ as a first-order confluent and terminating rewrite system. An implementation in \Dedukti{} is available on \url{https://gitlab.crans.org/geran/dedukti-level-implementation}.



\section{Universe representation in impredicative hierarchy} \label{sec-level-transformation}

In this section, we study the $\prodrule$ operators and its interaction with $\max$ and 
the successor and  establish semantic equalities that permits to simplify the levels
in order to find a set of sublevels for the desired representation.

\subsection{Levels as maximum}

The very first step is to show that any level can be express as a maximum of levels 
that do not contain any $\max$, that is the principle of our idea of representation.
The succesor can be distributed over $\max$, the two next propositions show
how to distribute $\prodrule$ over $\max$.

\begin{restatable}{proposition}{propmaxright} \label{prop-max_right}
	For all $u, v, w \in \Term$, $\prodrule[u, \max[v, w]] \eqTerm \max[\prodrule[u, v], \prodrule[u, w]]$.
\end{restatable}

\begin{restatable}{proposition}{propmaxleft} \label{prop-max_left}
	For all $u, v, w \in \Term$,
	$\prodrule[\max[u, v], w] \eqTerm \max[\prodrule[u, w], \prodrule[v, w]]$.
\end{restatable}

Then, any level can then be expressed as a maximum of levels without $\max$.
Note that for this, we consider that $\max$ takes a set of levels as argument.
We obtain this theorem.

\begin{theorem} \label{th-distr_rul_max}
	For all $t \in \Term$, there exists $u_1, \ldots, u_n$ in the grammar 
	$t \coloneqq 0 \mid s(t) \mid \prodrule[t, t]$ such that
	$t \eqTerm \max[u_1, \ldots, u_n]$.
\end{theorem}

\subsection{Simplification of the levels}

We can now focus on levels without maximum. The uniqueness property sought for the representation requires the subterms to be very basic, and then search to simplify
the levels.

The main issue is $\prodrule$: its asymmetry complicates its interaction with other symbols. The previous equalities show how to remove the interaction between $\prodrule$ and $\max$, now, we will study the interactions between $\prodrule$ and the other symbols. 
The goal is to restrict the localisation of the $\prodrule$ symbols to specific parts of the levels in order to understand and control their influence on the levels semantic.

Firstly, we recall these equalities that are direct consequences of the
semantic of $\prodrule$. They permit to deal with $0$ and the successor.

\begin{proposition} \label{prop-r_def}
For all $u, v \in \Term$, $\prodrule[u, 0] \eqTerm 0$ and
$\prodrule[u, s(v)] = \max[u, s(v)]$.
\end{proposition}

And we show how to remove $\prodrule$ symbol in second argument of
$\prodrule$.

\begin{restatable}{proposition}{proprright} \label{prop-r_right}
	For all $u, v, w \in \Term$, 
	$\prodrule[u, \prodrule[v, w]] \eqTerm \max[\prodrule[u, w], \prodrule[v, w]]$.
\end{restatable}

Thus, we can consider that the second argument of a $\prodrule$ is always
a variable. It is more complicated to directly enforce the form of its
first arguments, but we can obtain one restriction by distributing
the successor over the $\prodrule$. However, we cannot do it as directly
as we distribute the successor over the $\max$, as shown in the 
next example.

\begin{example}
	We show that $s(\prodrule[y, x]) \not\eqTerm \prodrule[s(y), s(x)]$
	by considering a valuation $\sigma$ such that $\sigma(x) = 0$ and 
	$\sigma(y) = 1$.
\end{example}

\begin{restatable}{proposition}{propdistrplus} \label{prop-distr_plus_r}
	For all $u, v, w \in \Term$, $s(\prodrule[v, w]) \eqTerm \max[s(w), \prodrule[s(v), w]]$.
\end{restatable}

Finally, all of these propositions leads to this grammar restriction.

\begin{theorem} \label{th-first_grammar}
	For all $t \in \Term$, there exists $u_1, \ldots, u_n$ in the grammar
	$t \coloneq s^k(x) \mid s^k(0) \mid \prodrule[t, x]$ such that
	$t \eqTerm \max[u_1, \ldots, u_n]$.
\end{theorem}

\begin{remark}
	For all $t$ in this grammar, there exists $x_1, \ldots, x_n \in \X$,
	and $v = s^k(0)$ or $v = s^k(x)$ such that
	$t = \prodrule[\prodrule[\prodrule[\cdots\prodrule[v, x_1], x_2)\cdots)],x_{n - 1}],x_n]$.
\end{remark}

\subsection{Introducting new levels}

Here, we continue the simplification of the levels under $\max$ in order to 
find simple enough terms to reach the uniqueness property. Indeed, the terms of the 
grammar of \cref{th-first_grammar} are still not enough.

\begin{example} \label{ex-a_without}
	Let us consider $t = \max[\prodrule[x, y], \prodrule[y, x]]$ Then,
	$t \eqTerm \max[x, y]$.
\end{example}

To better understand the issue, we should have a deep understanding of the semantic of
$\prodrule[x, y]$. It means that we always consider the value of $y$, but we only consider 
the value of $x$ if $y$ is not zero: the position of the variables are essential in the levels. But, taking into account $\imax[y, x]$ leads to offset this, and the value of $x$ and $y$
are always considered.

Here, we come to a solution that consists to introduce new symbols that permits to represent the current levels, without taking into account the order of the variables.

\begin{definition}[Sublevels] \label{def-sublevel}
	We define new symbols $\A$ and $\B$. $\A$ takes as argument a set of variables,
	an integer and a variable, and $\B$ takes as argument a set of variables
	and an integer. They have the following semantic.
	\begin{align*}
		\valuation{\A[E, x, S]}{\sigma} &= \begin{cases}
			0 & \text{if there exists $y \in E$ such that $\sigma(y) = 0$}\\
			\valuation{x}{\sigma} + S & \text{else}
		\end{cases}
		\\
		\valuation{\B(E, S)}{\sigma} &= \begin{cases}
			0 & \text{if there exists $y \in E$ such that $\sigma(y) = 0$}\\
			S & \text{else}
		\end{cases}
	\end{align*}
\end{definition}

In the next proposition, we show that any level of the grammar of 
\cref{th-first_grammar} can be written as a maximum of terms of this new grammar.

\begin{restatable}{proposition}{propassym} \label{prop-assym_a}
	If $t = \prodrule[\prodrule[\prodrule[\cdots\prodrule[s^k(y), x_1])\cdots)],x_{n - 1}],x_n]$, then 
	\[
	t \eqTerm 
	\max[
	\A[\emptyset, x_n, 0],
	\A[\set{x_n}, x_{n - 1}, 0], 
	\ldots, 
	\A[\set{x_2, \ldots, x_n}, x_n, 0],
	\A[\set{x_1, \ldots, x_n}, y, k]
	].
	\]
	And if $t = \prodrule[\prodrule[\prodrule[\cdots\prodrule[s^k(0), x_1])\cdots)],x_{n - 1}],x_n]$,
	then 
	\[
	t \eqTerm 
	\max[
	\A[\emptyset, x_n, 0],
	\A[\set{x_n}, x_{n - 1}, 0], 
	\ldots, 
	\A[\set{x_2, \ldots, x_n}, x_n, 0],
	\B[\set{x_1, \ldots, x_n}, k]
	].
	\]
\end{restatable}

The two equivalences are quite similar, the difference being in the very last subterm
of the $\max$ which is a $\A$ in the first case (we have to take into account $s^k(y)$
that is to say $k + y$) and a $\B$ in the second one ($k$ is taken into account with
a $\B$).

The sublevels that we search for our representation are elements of this grammar.
We just have two slightly modifications, two simplifications to make.
The first one is illustrated by this example.

\begin{example}
	With $t_1 = \A[\emptyset, x, 0]$ and $t_2 = \A[\set{x}, x, 0]$, we have
	$t_1 \eqTerm t_2$ since for all substitution $\sigma$, $\valuation{t_1}{\sigma} = \sigma(x) = \valuation{t_2}{\sigma}$.
\end{example}

The issue here is the fact that the second argument of a $\A$ symbol does not necesarily
appear in its first argument. This creates equalities as shown in the next proposition.

\begin{restatable}{proposition}{propavariable} \label{prop-a_variable}
	Let $x \in \X$, $E \subset \X \setminus \set{x}$ and $S\in \N$. Then
	\[
	\A[E, x, S] \eqTerm \max[\A[E \cup \set{x}, x, S], \B[E, S]].
	\]
\end{restatable}

Then, we will only consider elements of the form $\A[E, x, S]$ such that $x \in E$.

The second modification is related to the representation of $0$. Indeed, for all 
$E \subset \X$, $\B[E, 0] \eqTerm 0$, and we should then keep at most one of them.
However, since we already have $0 \eqTerm \max[\emptyset]$, we can remove all of them.

And we end up with this set of sublevels which permits to express any level of $\Term$.

\begin{definition}
	We denote by $\SubTerm$ the set of the sublevels that check the following conditions.
	\begin{enumerate}
		\item $\A[E, x, S] \in \SubTerm \iff x \in E$,
		\item $\B[E, S] \in \SubTerm \iff S > 0$. 
	\end{enumerate}
\end{definition}

\begin{theorem} \label{th-sublevels}
	Let $t \in \Term$. Then there exists $u_1, \ldots, u_n \in \SubTerm$ such
	that $t \eqTerm \max[u_1, \ldots, u_n]$. 
\end{theorem}

\begin{remark}
By convenience, we will note $u_i \in t$ if there exists $u_1, \ldots, u_n \in \SubTerm$ such
that $t \eqTerm \max[u_1, \ldots, u_n]$.
\end{remark}


\section{A minimal representation}
\label{sec-uniqueness}

In the previous section, we find a set $\SubTerm$ and show that any level can be represented
as a maximum of elements of $\SubTerm$. The goal of this one is to show that any level
has a minimal representation as maximum of elements of $\SubTerm$  and that this representation
is unique. Intuitevely, the sublevels of a minimal representation should be incomparable 
(else one of the sublevels could be removed).

\begin{definition}[Minimal representation] \label{def-normal}
	Let $t$ be a term. We say that $t$ is a minimal representation if and
	only if there exists $u1 \ldots, u_n \in \SubTerm$ such
	that 
	\begin{enumerate}
		\item $t = \max[u_1, \ldots, u_n]$,
		\item forall $i \neq j$, $u_i$ and $u_j$ are incomparable.
	\end{enumerate} 
	We denote by $\NFTerm$ the set of the minimal representations. 
\end{definition}

Of course, any level $t$ has a minimal representation. We can express
$t$ as a maximum of elements of $\SubTerm$ (by \cref{th-sublevels}) and we
remove elements while thy are not incomparabe. The challenging part 
is its uniqueness. To show it, we study the core of the definition 
of a minimal representation: the sublevel comparison.

\begin{restatable}[Sublevels comparison]{theorem}{thcomparisonsublevel}\label{th-comparison_sublevel}
	Elements of $\SubTerm$ are compared as follows.
	\begin{gather}
		\A[E, x, S] \not\leqTerm \B[F, K]\\
		\B[E, S] \leqTerm \B[F, K] \iff F \subset E \land S \leqslant K\\
		\B[E, S] \leqTerm \A[F, x, K] \iff (F \subset E \land S \leqslant K + 1)\\
		\A[E, x, S] \leqTerm \A[F, y, K] \iff F \subset E \land x = y \land S \leqslant K
	\end{gather}
\end{restatable}

As a corollary, we get that the sublevel equivalence is a syntactic equality, which is quite
natural; the uniqueness property would be impossible otherwise. 

\begin{corollary} \label{th-sub_uniqueness}
	Let $t_1, t_2 \in \SubTerm$. Then $t_1 \eqTerm t_2 \iff t_1 = t_2$.
\end{corollary}

And we can show the uniqueness of the minimal representation. 
First, we show that two equivalent minimal representations have the same
$\A$. 

\begin{proposition} \label{prop-normal_a}
	Let $t_1, t_2 \in \NFTerm$ such that $t_1 \eqTerm t_2$.
	Then
	\[
	\A[E, x, S] \in t_1 \iff \A[E, x, S] \in t_2.
	\]
\end{proposition}
\begin{proof}
	Let us note 
	\begin{gather*}
		t_1 = \max[\A[E_0, x_0, S_0], \ldots, \A[E_n, x_n, S_n], 
		\B[G_0, T_0], \ldots, \B[G_p, T_p]]\\
		t_2 = \max[\A[F_0, y_0, K_0], \ldots, \A[F_m, y_m, K_m],
		\B[H_0, L_0], \ldots, \B[H_q, L_q]].
	\end{gather*}
	Let $\A[E, x, S]$ be an sublevel of $t_1$. We consider $\sigma$ such that
	\[
	\sigma(y) = \begin{cases*}
		\max[S_0, \ldots, S_n, K_0, \ldots, K_m, T_0, \ldots, T_p, L_0, \ldots, L_q] + 1
		& if $y = x$\\
	 1 & if $y \in E \setminus \set{x}$\\
	 0 & else
	\end{cases*}
	\]
	We have $\valuation{t_1}{\sigma} = \valuation{\A[E, x, S]}{\sigma} = S + \sigma(x)$ and then 
	$\valuation{t_2}{\sigma} = S + \sigma(x)$. Then, there exists 
	$\A[F, y, K]$ in $t_2$ such that $F \subset E \cup \set{x} = E$ 
	(else $F$ contains a variable $z$ such that $\sigma(z) = 0$) and
	$\sigma(y) + K = \sigma(x) + S$ or there exists $\B[F, K]$ in $t_2$
	such that $K = \sigma(x) + S$. Since  
	$\sigma(x) > \max[S_0, \ldots, S_n, K_0, \ldots, K_m]$, we deduce that
	it is the first case and $y = x$.
	
	Then, there is $(F, x, S)$ in $t_2$ with $F \subset E$ and $\sigma(y) + K = \sigma(x) + S$. If $F \subsetneq E$, then by the same reasoning, we show that
	there exists $\A[G, x, S] \in t_1$ with $G \subset F \subsetneq E$. But,
	by \cref{def-normal}, it is impossible to have $\A[E, x, S]$ and $\A[G, x, S]$ in 
	$t_1$ with $G \subset E$ since they are comparable.
	
	Then $E = F$ and $\A[E, x, S]$ is also an element of $t_2$.
\end{proof}

And we show the same for the $\B$.

\begin{proposition} \label{prop-normal_b}
	Let $t_1, t_2 \in \NFTerm$ such that $t_1 \eqTerm t_2$.
	Then
	\[
	\B[E, S] \in t_1 \iff \B[E, S] \in t_2.
	\]
\end{proposition}
\begin{proof}
	Let us note 
	\begin{gather*}
		t_1 = \max[\A[E_0, x_0, S_0], \ldots, \A[E_n, x_n, S_n], 
		\B[G_0, T_0], \ldots, \B[G_p, T_p]]\\
		t_2 = \max[\A[F_0, y_0, K_0], \ldots, \A[F_m, y_m, K_m],
		\B[H_0, L_0], \ldots, \B[H_q, L_q]].
	\end{gather*}
		We show the result by induction on $E$. Let $\B[E, S]$ be a sublevel of $t_1$. 
		If $E = \emptyset$, we consider $\sigma$ the zero function. Then, 
		$\valuation{t_1}{\sigma} = S$, hence $ \valuation{t_2}{\sigma} = S$.
		Since $S > 0$, it follows that $\B[\emptyset, S]$ is a sublevel of $t_2$.
		
		In the induction case, we consider $\sigma$ such that 
		$\sigma(x) = 1$ if $x \in E$ and $\sigma(x) = 0$ otherwise,
		hence $\valuation{t_1}{\sigma} = S$. 
		Then, $\valuation{t_2}{\sigma} = S$ and since $S > 0$, 
		there exists $\A[F, x, K]$ in $t_2$ such that $F \subset E$ and 
		$\sigma(x) + K = S$ or there exists $\B[F, S]$ in $t_2$ such that $F \subset E$ 
		and $K = S$.
	
	In the first case, we have $x \in F \subset E$, then $\sigma(x) = 1$ and 
	$K = S - 1$. Then, by \cref{prop-normal_a}, $\A[F, x, S - 1] \in t_1$ which 
	is impossible by \cref{def-normal} since it would be comparable with 
	$\B[E, S] \in t_1$.
	
	Then, we have $\B[F, S] \in t_2$. If $F \subsetneq E$, we apply the induction
	hypothesis and obtain $\B[F, S] \in t_1$, impossible because it would be
	comparable with $\B[E, S]$.

	Then $E = F$ and $\B[E, S]$ is also an element of $t_2$.
\end{proof}

We immediately obtain that equivalence of minimal representations is the
syntactic equality.

\begin{proposition} \label{th-uniqueness}
	For all $t_1, t_2 \in \NFTerm$, $t_1 \eqTerm t_2 \iff t_1 = t_2$.
\end{proposition}

And finally, we obtain the main theorem: the existence and uniquennes of
a minimal representation for each level. Before, we show the intuitive property
that the minimal representation of a maximum of sublevels is formed with 
elements of these sublevels.

\begin{proposition} \label{prop-expand_subterms}
	For all $u_1, \ldots, u_n \subset \SubTerm$, there exists 
	a unique $\set{v_1, \ldots, v_m} \subseteq \set{u_1, \ldots, u_n}$ such that
	$\max[u_1, \ldots, u_n] = \max[v_1, \ldots, v_m]$ and
	$\max[v_1, \ldots, v_m] \in \NFTerm$.
\end{proposition}

\begin{theorem}[Representation]\label{th-exist_normal}
	For all $t \in \Term$, there exists a unique 
	$\set{u_1, \ldots, u_n} \subset \SubTerm$ such that
	$t \eqTerm \max[u_1, \ldots, u_n]$.
	We say that $\max[u_1, \ldots, u_n]$ is the minimal representation
	of $t$ and we denote it as $\repr[t]$.
\end{theorem}

Having a unique minimal representation is useful to have a faithful embedding of
$\Term$ in the $\lambdapimod$, which is done in \cref{sec-implementation}. Moreover,
this representation also gives us an easy way to compare two terms. 
Indeed, a sublevel can be compared to a level using the representation.

\begin{restatable}{lemma}{propindependence} \label{prop-independence}
	Let $u, v_1, \ldots, v_n \in \SubTerm$. Then $u \leqTerm \max[v_1, \ldots, v_n]$ if 
	and only if there exists $i$ such that $u \leqTerm v_i$.
\end{restatable}
	\begin{proof}
	The reverse implication is trivial. We show the direct implication by 
	contraposition. We suppose that 
	forall $i$, $u \not\leqTerm v_i$. 
	
	If $u = \B[E, S]$, we consider $\sigma$ such that $\sigma(x) = 1$ if $x \in E$ and 
	$0$ otherwise. Then, forall $v_i$, we have either
	\begin{itemize}
		\item $v_i = \A[F, x, K]$ or $v_i = \B[F, K]$ and $F \not\subset E$ hence
		$\valuation{v_i}{\sigma} = 0 < S + M + 2 = \valuation{u}{\sigma}$,
		\item or $v_i = \A[F, x, K]$ with $F \subseteq E$ and $K < S - 1$ hence
		$\valuation{v_i}{\sigma} = K + 1 < S = \valuation{u}{\sigma}$,
		\item or $v_i = \B[F, K]$ with $K < S - 1$ hence
		$\valuation{v_i}{\sigma} = K < S = \valuation{t}{\sigma}$.
	\end{itemize}
	Then $u \not\leqTerm \max[v_1, \ldots, v_n]$.
	
	If $u = \A[E, x, S]$, then, each $v_i$ is of the form $\A[F_i, x_i, K_i]$ or
	$\B[F_i, K_i]$, we consider $M$ the maximum of these $K_i$ and $\sigma$ such that 
	$\sigma(x) = M + 2$, $\sigma(y) = 1$ if $y \in E \setminus \set{x}$ and 
	$0$ otherwise. Then, forall $v_i$, we have either
	\begin{itemize}
		\item $v_i = \B[F, K]$ hence
		$\valuation{v_i}{\sigma} \leqslant K < S + M + 2 = \valuation{u}{\sigma}$,
		\item or $v_i = \A[F, y, K]$ and $F \not\subset E$ hence
		$\valuation{v_i}{\sigma} = 0 < S = \valuation{u}{\sigma}$,
		\item or $v_i = \A[F, y, K]$ with $F \subseteq E$ and $x \neq y$, hence
		$\valuation{v_i}{\sigma} = K + 1 < S + M + 2 = \valuation{u}{\sigma}$,
		\item or $v_i = \A[F, x, K]$ with $F \subseteq E$ and $K < S$, hence
		$\valuation{v_i}{\sigma} = K + M + 2 < S + M + 2 = \valuation{u}{\sigma}$.
	\end{itemize}
	Then $u \not\leqTerm \max[v_1, \ldots, v_n]$.
\end{proof}

And we use this lemma to compare two level, for instance by comparing each sublevel
of the minimal representation of the firt one to the second one. More generally,
two representation are compared the following way.

\begin{theorem} \label{th_comparison}
	Let $u_1, \ldots, u_n, v_1, \ldots, v_m \in \SubTerm$. Then, 
	$\max[u_1, \ldots, u_n] \leqTerm \max[v_1, \ldots, v_m]$ if and only if forall $i$, 
	there exists $j$ such that $u_i \leqTerm v_j$.
\end{theorem}

One can note that the \cref{prop-independence} gives us a new proof of the uniqueness property stated in \cref{th-uniqueness}.
\begin{proof}
	Let $u = \max[u_1, \ldots, u_n]$ and $v = \max[v_1, \ldots, v_m]$ be two minimal representations such that $u \eqTerm v$. We want to show that forall $i$, there exists $j$
	such that $u_i = v_j$.

	We have $u_i \leqTerm u \leqTerm v$, hence by \cref{prop-independence}, there exists $v_j$ such that $u_i \leqTerm v_j$. In the same way, there exists $u_k$ such that 
	$v_j \leqTerm u_k$. Then, by \cref{def-normal}, $i = k$ (because $u_1, \ldots, u_n$ are incomparable), and then $u_i \eqTerm v_j$ hence $u_i = v_j$ by \cref{th-sub_uniqueness}.
\end{proof}

This shows that there is a link between the \cref{prop-independence} and the uniqueness
property. In fact, this lemma should be understood as an \emph{independence} lemma.
Indeed, if we consider $\max[u_1, \ldots, u_n]$ as a linear combination of $u_1, \ldots, u_n$,
then this lemma states that the only way to be smaller than a linear combination is to depend 
and be smaller than of one of the element of this combination.

This analogy provides a new point of view on our work: $\SubTerm$ is a \enquote{linearly independent} family (uniqueness of the minimal representation) which generates at least $\Term$
(for instance, $\max[\A[\set{x}, y, 0]]$ or $\max[\B[\set{x}, 1]]$ are not equivalent to any level).


\section{A rewriting system for this universe representation}
\label{sec-implementation}

This section is dedicated to the implementation of this representation in the $\lambdapimod$. 

\subsection{Basic tools}

Here, we define the very basic term that we will used. The booleans and the natural
numbers, to begin with, are necesary.

\begin{definition}[Booleans]
	We define a type $\Bool$, with constructors $\true \colon \Bool$, $\false \colon \Bool$ and functions
	$\ands \colon \Bool \to \Bool$, $\ors \colon \Bool \to \Bool$ and 
	$\nots \colon \Bool \to \Bool$.
	$\Bool$ is interpreted as the Booleans, $\true$, $\false$, $\ands$, $\ors$ and $\nots$
	as $\top$, $\bot$, the conjonction, the disjonction and the negation.
\end{definition}

\begin{definition}[Natural numbers]
We define a type $\Nat$, with constructors $\zero \colon \Nat$, $\s \colon \mathrm{\Nat} \to \Nat$ and  functions $\plus \colon \Nat \to \Nat \to \Nat$, $\leqs{\Nat} \colon \Nat \to \Nat \to \Bool$, $\eqs{\Nat} \colon \Nat \to \Nat \to \Bool$, $\lqs{\Nat} \colon \Nat \to \Nat \to \Bool$ (with infix notation) and $\maxN \colon \Nat \to \Nat \to \Nat$.
$\Nat$ is interpreted as $\N$, $\zero$ as $0$, $\s$, $\plus$, $\leqs{\Nat}$, $\eqs{\Nat}$ and $\maxN$ 
as $0$, $s$, $+$, $\leq$, $=$ and $\max$.
\end{definition}

Moreover, we define a \emph{if-then-else} structure. It is not really necesary,
but will be very convenient to facilitate the writing of some rules.

\begin{definition}
	Let $T$ be a type. We define the function $\iteop_T \colon \Bool \to T \to T \to T$
	such that $\forall u, v \in T$, 
	$\app{\app{\app{\iteop_T}{\true}}{u}}{v} \rewrites u$ and 
	$\app{\app{\app{\iteop_T}{\false}}{u}}{v} \rewrites v$.
	For convenience reasons, we will denote $\app{\app{\app{\iteop_T}{b}}{u}}{v}$ by
	$\app{\app{\app{\app{\app{\ifs}{b}}{\thens}}{u}}{\elses}}{v}$.
\end{definition}

And finally, we show how to define a type of sets for all ordered types. It will be used
for set of natural numbers (to define the sublevels), but also for set of 
sublevels (to define the representations).

\begin{definition}[Sets] \label{prop_set}
	Let $T$ be a type equipped with a total order function $\order{T} \colon T \to T \to \Bool$.
	Then, we define a type $\Set{T}$ corresponding to the finite set of elements
	of $T$ with a constructor $\nil{T} \colon \Set{T}$ and the functions (all with infix notation)
	$\add{T} \colon \Set{T} \to T \to \Set{T}$,  
	$\union{T} \colon \Set{T} \to \Set{T} \to \Set{T}$,
	$\included{T} \colon T \to \Set{T} \to \Bool$,
	$\sub{T} \colon \Set{T} \to \Set{T} \to \Bool$ and
	$\eqs{\Set{T}} \colon \Set{T} \to \Set{T} \to \Bool$.
	
	$\Set{T}$ is interpreted as the set of finite subsets of $T$,
	$\nil{T}$ as the empty set, $\add{T}$ as the function that adds an element to
	a set, $\union{T}$, $\included{T}$, $\sub{T}$ and $\eqs{\Set{T}}$ as
	$\cup$, $\in$, $\subseteq$ and $=$.

	For convenience reasons, we will denote the term 
	$\nil{T} \add{T} x_1 \add{T} \cdots \add{T} x_n$
	by $\setS{T}{x_1, \ldots, x_n}$. In particular, 
	$\setS{T}{x}$ will represent a singleton.
\end{definition}

	$\Set{T}$ is implemented as a sorted list of elements of $T$ with the
	constructors $\nil{T} \colon \Set{T}$ and $\cons{T} \colon T \to \Set{T} \to \Set{T}$,
	and the function $\add{T}$ adds an element to a list while keeping the uniqueness and 
	order properties. Then, a set will be build through $\nil{T}$ and $\add{T}$ and
	the constructor $\cons{T}$ will only be used in patterns of rewrite rules.

\begin{definition}[Set order] \label{prop-set_order}
	Let $T$ be a type with a total order function $\order{T}$. Then, we define a 
	total order $\order{\Set{T}}$ by considering the lexicographic order on 
	sorted word of $T$. Besides, we define the corresponding strict total order 
	$\strictOrder{\Set{T}}$.
\end{definition}

We do not give the rules for these elements since they are quite basic.
However, the \Dedukti{} implementation is available if necesary.


\subsection{Level encoding}

In order to implement the sets of sublevels and to compare two sublevels, we should be able to compare and sort level variables. That is why we use a deep encoding where each variable is encoded as a natural number. We denote by $\gamma \colon \X \to \N$ a bijection that associates each variable to a natural number.

\begin{definition} \label{def-basic_levels}
	We define $\Level$, the type of the levels together with the constructors 
	$\zeroL \colon \Level$, $\sL \colon \Level \to \Level$, $\maxL \colon \Level \to \Level \to \Level$, $\ruleL \colon \Level \to \Level \to \Level$ and $\varL \colon \Nat \to \Level$.
We define inductively a translation function $\translation{\cdot} \colon \Term \to \Level$
with 
\begin{gather*}
\translation{0} = \zeroL \qquad \translation{s(t)} = \app{\sL}{t} \qquad
\translation{x} = \app{\varL}{\translation{\gamma(x)}_{\Nat}}\\
\translation{\max[u, v]} = \app{\app{\maxL}{\translation{u}}}{\translation{v}} \qquad
\translation{\prodrule[u, v]} = \app{\app{\ruleL}{\translation{u}}}{\translation{v}}
\end{gather*}
\end{definition}

\begin{definition}
	We define $\SLevel$, the type of the sublevels, with the constructors 
	$\As \colon \NatSet \to \Nat \to \Nat \to \SLevel$ and 
	$\Bs \colon \NatSet \to \Nat \to \SLevel$. We define a translation function $\translation{\cdot}_{\Sub} \colon \SubTerm \to \SLevel$ 
	by $\translation{\A[E, x, S]}_{\Sub} = 
	\app{\As}{\app{\translation{E}}{\app{\translation{x}}{\translation{S}}}}$
	and
	$\translation{\B[E, S]}_{\Sub} = \app{\Bs}{\app{\translation{E}}{\translation{S}}}$.
\end{definition}

In order to define $\Set{\SLevel}$, we need a total order on $\SLevel$. We consider
the order where all the $\As$ are before the $\Bs$, and the $\As$ (respectively) the $\Bs$ are sorted according to the lexicographic order:
\begin{itemize}
	\item $\app{\As}{\app{E}{\app{x}{S}}} \leqslant \app{\Bs}{\app{F}{K}}$,
	\item $\app{\As}{\app{E}{\app{x}{S}}} \leqslant \app{\As}{\app{\translation{F}}{\app{\translation{y}}{\translation{K}}}}$ is the lexicographic
	order between $(E, x, S)$ and $(F, y, K)$ (which is a total order since 
	$\Nat$ and $\Set{\Nat}$ are equipped with a total order by \cref{prop-set_order}),
	\item $\app{\Bs}{\app{E}{S}} \leqslant \app{\Bs}{\app{F}{K}}$ is the lexicographic order between $(E, S)$ and $(F, K)$.
\end{itemize}

\begin{definition}[Total order on sublevels] \label{prop-total_order_sublevel}
	We define a total $\order{\SLevel}$ on $\SLevel$ (and its corresponding 
	strict order $\strictOrder{\SLevel}$).
Indeed, we consider 
Hence these rewrite rules.
\begin{align*}
(\app{\As}{\app{E}{\app{x}{S}}}) \order{\SLevel} (\app{\Bs}{\app{F}{K}}) &\rewriteDef \true\\
(\app{\As}{\app{E}{\app{x}{S}}}) \order{\SLevel} (\app{\As}{\app{F}{\app{y}{K}}}) 
&\rewriteDef \app{\ors}{(E \strictOrder{\Set{\Nat}} F)}{}\\ 
& \phantom{\rewriteDef{}} \app{\ands}
			{\app{(E \eqs{\Set{\Nat}} F)}
			({(\app{\ors}
				{\app{(x \lqs{\Nat} y)}
				{{(\app{\ands}{\app{(x \eqs{\Nat} y)}{(S \leqs{\Nat} K)}})}}})}				
			})	
	\\
\app{\Bs}{\app{E}{S}} \order{\SLevel} \app{\Bs}{\app{F}{K}} &\rewriteDef 
	\ors[
		E \strictOrder{\Set{\Nat}} F,
		\ands[
			E \eqs{\Set{\Nat}} F, 
			S \leqs{\Nat} K
		]
	]
\end{align*}
\end{definition}

And we can then define the representations.

\begin{definition}
	We define a function $\maxS \colon \SLSet \to \Level$ that embeds a set 
	of $\SLevel$ into a $\Level$ and a translation function 
	$\translation{\cdot}_{\NF} \colon \NFTerm \to \Level$ defined by 
	$\translation{\max[u_1, \ldots, u_n]}_{\NF} = 
	\maxS[\setS{T}{\translation{u_1}_{\Sub}, \ldots, \translation{u_n}_{\Sub}}]$.
\end{definition}

Now, we have defined all the types and the elements of $\NFTerm$ have translations
in $\lambdapimod$. The next step is to provide rewrite rules that transforms a level
into its minimal representation. The cases of $\zeroL$ and of variables are easy.
\[
\zeroL \rewriteDef \maxS[\nil {\SLevel}] \qquad
\varL[x] \rewriteDef \maxS[\setS{\SLevel}{\As[\setS{\Nat}{x}, 0, x]}]
\]

For the other cases, and in particular for the $\max$, the sublevel comparison will be 
useful, then we define ite.

\begin{definition}[Sublevels comparison]
We define $\leqs{\SLevel} \colon \SLevel \to \SLevel \to \Bool$
interpreted as $\leqTerm$ 
with these rewrite rules.
\begin{align*} 
	\As[E, x, S] \leqs{\SLevel} \Bs[F, K] &\rewriteDef 
	\false & \tag{1}\\
	\Bs[E, S] \leqs{\SLevel} \Bs[F, K] &\rewriteDef 
	\ands[F \sub{\Nat} E, S \leqs{\Nat} K] & \tag{2}\\
	\Bs[E, \sL[S]] \leqs{\SLevel}  \As[F, y, K] &\rewriteDef
	\ands[F \sub{\Nat} E, S \leqs{\Nat} K] & \tag{3}\\
	\As[E, x, S] \leqs{\SLevel} \As[F, y, K] &\rewriteDef
	\ands[F \sub{\Nat} E, \ands[x \eqs{\Nat} y, S \leqs{\Nat} K]]
	& \tag{4}
\end{align*}
\end{definition}

\begin{proposition}
	Let $u, v \in \SubTerm$. Then, 
	$u \leqTerm v \iff \translation{u} \leqs{\SLevel} \translation{v} \rewrites^* \true$.
\end{proposition}
\begin{proof}
	Each rule ($i$) corresponds to the case $i$ of the \cref{th-comparison_sublevel}.
\end{proof}


And we can give the remaining rules.

\subsection{The successor} \label{subsec-succ}

We show how to compute the minimal representation of $s(\max(u_1, \ldots, u_n))$. 

\begin{definition}
	We define $\succSub$ on the sublevels by $\succSub(\A[E, x, S]) = \A[E, x, S + 1]$ and $\succSub(\B[E, S]) = \B[E, S + 1]$. Then $\succSub(u) \eqTerm s(u)$.
\end{definition}

\begin{proposition}\label{prop-compute_succ}
	$\repr[s(\max[\emptyset])] = \max[\B[\emptyset], s(0)]$ and for all $n > 0$ and 
	$t = \max[u_1, \ldots, u_n] \in \NFTerm$,
	$\repr[s(t)] = \max[\succSub(u_1), \ldots, \succSub(u_n)]$.
	
\end{proposition}

We create a function $\sSL \colon \SLSet \to \SLSet$ corresponding to $\succSub$.
\begin{gather*}
	\sSL[\nil{\SLevel}] \rewriteDef \nil{\SLevel}  \qquad
	\sSL[\Bs[E, S] \cons{\SLevel} q] \rewriteDef \sSL[q] \add{\SLevel} \Bs[E, \s[S]] \\
	\sSL[\As[E, x, S] \cons{\SLevel} q] \rewriteDef \sSL[q] \add{\SLevel} \As[E, x, \s[S]]
\shortintertext{And we compute the minimal representation of a successor according to \cref{prop-compute_succ}.}
	\sL[\maxS[\nil{\SLevel}]] \rewriteDef \maxS[\Bs[\nil{\SLevel}, \s[0]]] \qquad 
\sL[\maxS[u \cons{\SLevel} q]] \rewriteDef \maxS[\sSL[u \cons{\SLevel} q]]
\end{gather*}


\begin{proposition} \label{prop-succ_sound}
	For all $t \in \NFTerm$,  
	$\sL[\translation{t}_{\NF}] \rewrites^* \translation{\repr[s(t)]}_{\NF}$.
\end{proposition}

\subsection{The maximum} \label{subsec-max}

For the maximum, we define $\max_{\Sub}$, that adds a sublevel to a minimal representation.

\begin{definition}
	We define inductively $\maxSub \colon \setPart[\SubTerm] \times \SubTerm \to \setPart[\SubTerm]$ inductively defined by 
	$\maxSub[\emptyset, u] = \set{u}$ and
	\[
	\maxSub[\set{u_1, \ldots, u_n} \cup \set{u}, v] = \begin{cases*}
		\set{u_1, \ldots, u_n, u}        & if $v \leqTerm u$,\\
		\set{u_1, \ldots, u_n, v}        & if $u \leqTerm v$,\\
		\maxSub[\set{u_1, \ldots, u_n}, v] \cup \set{u} & else.
	\end{cases*}
	\]
\end{definition}

\begin{proposition} \label{prop-compute_max}
	For all minimal representation $t = \max[u_1, \ldots, u_n] \in \NFTerm$ and 
	$v \in \SubTerm$, $\repr[\max[u_1, \ldots, u_n, v]] = f(\set{u_1, \ldots, u_n}, v)$.
\end{proposition}

We implement $\maxSub$ as a function 
$\maxhelper \colon \SLSet \to \SLevel \to \SLSet$ with these rewrite rules.
\begin{align*} 
\maxhelper[\nil{\SLevel}, u] 
	&\rewriteDef \setS{\SLevel}{u}\\ 	
\maxhelper[u \cons{\SLevel} E , v] 
	&\rewriteDef \ifs v \leqs{\SLevel} u \thens E \add{\SLevel} u\\
	&\phantom{\rewriteDef{} } \elses \ifs u \leqs{\SLevel} v \thens E \add{\SLevel} v
	\elses \maxhelper[E, v] \add{\SLevel} u  
\end{align*}
And we compute the minimal representation for $\maxL$ according to 
\cref{prop-compute_max}.
\begin{align*}
	\maxL[\maxS[E], \maxS[\nil{\SLevel}]] &\rewriteDef \maxS[E]\\
	\maxL[\maxS[E], \maxS[u \cons{\SLevel} F]] &\rewriteDef
	\maxL[\maxhelper[E, u], \maxS[F]]
\end{align*}

\begin{proposition} \label{prop-max_sound}
	For all $t_1, t_2 \in \NFTerm$,  
	$\maxL[\translation{t_1}_{\NF}, \translation{t_2}_{\NF}] \rewrites^* \translation{\repr[\max(t_1, t_2)]}_{\NF}$.
\end{proposition}

\subsection{The rule} \label{subsec-max}

To begin, we study $\prodrule[u, v]$ where $u, v \in \SubTerm$.

\begin{proposition} \label{prop-rule_sub}
	Let $f(E, S)$ be either $\A[E, x, S]$ or $\B[E, S]$ and $g(F, K)$ be either 
	$\A[F, x, K]$ or $\B[F, K]$. Then 
	\[\prodrule[f(E,S), g(F, K)] \eqTerm \max[f(E \cup F, S), g(F, K)].
	\]	
\end{proposition}
\begin{proof}
	We note $u = f(E, x, S)$ and $v = g(F, K)$. Let $\sigma$ be a substitution. 
	\begin{itemize}
		\item If there exists $y \in F$ such that $\sigma(y) = 0$, then 
		$
		\valuation{\prodrule[u, v]}{\sigma} = 0 = 
		\valuation{\max[f(E \cup F, S), v]}{\sigma}.
		$
		\item Else,$\valuation{v}{\sigma} = K > 0$ and then
		$\valuation{\prodrule[u, v]}{\sigma} = \max[\valuation{u}{\sigma}, K]$.
		Besides, since $\sigma(y) \neq 0$ forall $y \in F$, there exists $y \in E \cup F$ such that
		$\sigma(y) = 0$ if and only if there exists $y \in E$ such that $\sigma(y) = 0$, hence 
		$\valuation{u}{\sigma} = \valuation{f(E \cup F, S)}{\sigma}$.
	\end{itemize}
	Hence the result.
\end{proof}
We implement it as a function $\ruleSL \colon \SLevel \to \SLevel \to \Level$. 
\begin{align*}
	\ruleSL[\As[E, x, S], \Bs[F, K]] 
	&\rewriteDef	\maxL[\maxS[\As[E \union{\Nat} F, x, S]], \maxS[\Bs[F, K]]]\\
	\ruleSL[\Bs[E, S], \Bs[F, K]] 
	&\rewriteDef	\maxL[\maxS[\Bs[E \union{\Nat} F, S]], \maxS[\Bs[F, K]]]\\
	\ruleSL[\Bs[E, S], \As[F, x, K]] 
	&\rewriteDef	\maxL[\maxS[\Bs[E \union{\Nat} F, S]], \maxS[\As[F, x, K]]]\\
	\ruleSL[\As[E, x, S], \As[F, y, K]] 
	&\rewriteDef	\maxL[\maxS[\As[E \union{\Nat} F, x, S]], \maxS[\As[F, y, K]]]\\
\end{align*}
Then, following the equalities $\prodrule[0, t] \eqTerm t$ and 
$\prodrule[t, 0] \eqTerm 0$, and \cref{prop-max_right,prop-max_left},
we define $\rulehelper \colon \SLevel \to \Level \to \Level$ and add these rewrite rules.
\begin{align*}
\ruleL[\maxS[\nil{\SLevel}], t]           &\rewriteDef t\\
\ruleL[\maxS[u \cons{\SLevel} q], t]      &\rewriteDef\maxL[\rulehelper[u, t], \ruleL[q, t]]\\
\rulehelper[u, \maxS[\nil{\SLevel}]]      &\rewriteDef \maxS[\nil{\SLevel}]\\
\rulehelper[u, \maxS[v \cons{\SLevel} q]] &\rewriteDef \maxL[\ruleSL[u, v], \rulehelper[u, q]]
\end{align*}

\begin{proposition} \label{prop-rule_sound}
		For all $t_1, t_2 \in \NFTerm$,  
	$\ruleL[\translation{t_1}_{\NF}, \translation{t_2}_{\NF}] \rewrites^* \translation{\repr[\prodrule(t_1, t_2)]}_{\NF}$.
\end{proposition}
\begin{proof}
	By \cref{prop-rule_sub,prop-max_sound}, for all $u, v \in \SubTerm$, 
	$\ruleSL[\translation{u}_{\Sub}, \translation{v}_{\Sub}] \rewrites^* 
	\translation{\repr[\prodrule[u, v]]}$.
	Then, by induction on $t \in \NFTerm$ and using \cref{prop-max_sound},
	we show that forall $u \in \SubTerm$,  
	$\rulehelper[\translation{u}_{\Sub}, \translation{t}_{\NF}] \rewrites^* 
	\translation{\repr[\prodrule[u, t]]}_{\NF}$.
	And finally, we show the result by induction on $t_2$ using \cref{prop-max_right,prop-max_left} 
	and the equivalences $\prodrule[0, t] \eqTerm t$ and $\prodrule[t, 0] \eqTerm 0$.
\end{proof}

\subsection{Implementing the substitution}

Since we use a deep encoding, $\beta$-reduction cannot be used for substitution. Then, we implement an substitution function.
First, we implement $\evalS \colon \SLevel \to \Nat \to \Nat \to \Level$ for the sublevels following the semantic given in the \cref{def-sublevel}.
\begin{align*}
\evalS[\Bs[E, S], y, n] 
& \rewriteDef  \ifs \ands[y \included{\Nat} E, n \eqs{\Nat} 0] \thens \maxS[\nil{\SLevel}]\\
& \phantom{\rewriteDef{}} \elses \maxS[\Bs[E \del{\Nat} y, S]]\\
\evalS[\As[E, x, S], y, n] 
&\rewriteDef  \ifs \ands[y \included{\Nat} E, n \eqs{\Nat} 0] \thens \maxS[\nil{\SLevel}]\\
& \phantom{\rewriteDef{}} \elses \ifs x \eqs{\Nat} y \thens \maxS[\Bs[\nil{\Nat}, S \plus n]]\\
& \phantom{\rewriteDef{}} \elses \maxS[\As[E \del{\Nat} y, x, S]] 
\end{align*}
Then, we create a function $\evalL \colon \Level \to \Nat \to \Nat \to \Level$
that evaluate a level using the fact that 
$\subst{\max[u_1, \ldots, u_n]}{\set{x \mapsto n}} = 
		\max[\subst{u_1}{\set{x \mapsto n}}, \ldots, \subst{u_n}{\set{x \mapsto n}}]$.
\[
\evalL[\nil{\SLevel}, y, n] \rewriteDef \nil{\SLevel}\qquad
\evalL[u \cons{\SLevel} q, y, n] \rewriteDef \maxL[\evalS[u, y, n], \evalL[q, y, n]]
\]

\begin{proposition}
	Let $t \in \Term$. Then, the normal form of $\translation{\subst{t}{\set{x \mapsto u}}}$ is 
	$\evalL[\translation{t}, \translation{u}]$. 
\end{proposition} 


\subsection{Properties of the rewrite system}

The rewrite system that we designed have strong properties. First, one can note that it
is does not use any higher-order rewrite rule, hence it can be implemented in a first
order system.

\begin{theorem}
	The rewrite system is confluent and strongly normalizing.
\end{theorem}
\begin{proof}
	The termination has been proved with two termination checkers, \TTTT{} \cite{TTT2} and 
	 SizeChangeTool \cite{SCT}, and the confluence with CSI \cite{CSI}.
\end{proof}

And of course, we show that it is sound relatively to the minimal representation.

\begin{theorem}[Soundness] \label{th-soundness}
	Let $t \in \Term$. Then, the normal form of $\translation{t}$ is 
	$\translation{\repr[t]}_{\NF}$. 
\end{theorem}
\begin{proof}
	We show that $t \rewrites^* \translation{\repr[t]}_{\NF}$ by induction on $t$: $\varL[x] \rewrites \maxS[\setS{\SLevel}{\As[\setS{\Nat}{x}, 0, x]}]$, $0 \rewrites \maxS{\nil{\SLevel}}$, 
	and we show the cases $s$, $\max$ and $\prodrule$ using \cref{prop-max_sound,prop-rule_sound,prop-succ_sound}.
	Since no rewrite rule can be applied to $\maxS$, the normal
	form of $\translation{t}$ is $\translation{\repr[t]}_{\NF}$.
\end{proof}

\cref{th-soundness,th-exist_normal} gives us that the translations of two equivalent levels
are convertible, and even more, they have the same normal form. In other words, this embedding
faithfully represent the level equivalences.

Besides, one could note some drawbacks of this embedding. First, we do not have a back translation
from $\lambdapimod \to \Level$. There are two main reason for this.
\begin{enumerate}
	\item $\As$ and $\Bs$ are not exact translations of $\A$ and $\B$. 
	\item $\NFTerm$ and $\Term$ are not equivalent.
\end{enumerate}
The first reason is linked to the restrictions that we added to $\A$ and $\B$. Here, it is possible
to write terms $\app{\Bs}{\app{E}{\zero}}$ or $\app{\As}{\app{E}{\app{x}{k}}}$ where $x$ is not an
element of $E$. In the same way, it is possible to write $\app{\maxS}{L}$ while two sublevels of
$L$ are comparable. 

A solution could be to add a dependent term as argument, to check these conditions. For instance,
$\As$ would have the type 
$(E \colon \Set{\Nat}) \to (x \colon \Nat) \to \app{\Prf}{(x \included{\Nat} E)} \to \Nat \to \SLevel$ where $\Prf \colon \Bool \to \mathtt{Type}$ represents the proof of a proposition. We
declare $\I \colon \app{\Prf}{\true}$ and we use $\app{\As}{\app{E}{\app{x}{\app{k}{\I}}}}$
using the fact that $x \included{\Nat} E$ reduces to $\true$ if and only if $x$ is an element
of $E$.

The second reason is not related to the embedding, but to the representation that we introduced.
Indeed, we already note that some minimal representations are not equivalent to any level
($\max[\A[\set{x}, y, 0]]$ or $\max[\B[\set{x}, 1]]$ as examples). Then, it could be
a good idea to find a characterization of the elements of $\NFTerm$ that actually correspond to
levels.


\section{Conclusion}

We introduced a new representation of the levels of the impredicative PTS where equivalent levels have the same representation. It provides us an easy procedure decision for the inequality problem in the 
$\imax$-sucessor algebra, and it permits us to get a sound encoding of these levels in the $\lambdapimod$, in the sense that equivalent levels have convertible translations. Moreover, this encoding corresponds to a first-order, confluent and strongly normalizing rewrite system, and in particular it permits to decide level equality. 

This encoding of the levels permits to encode $\ICC$ with universe polymorphism. Besides,
we still have to study how this encoding behaves well together with encodings of inductive types or cumulativity in order to get a better encoding of \Coq. The ideas mentioned at the end of
\cref{sec-implementation} are also interesting. In particular, the characterization of 
the elements of $\NFTerm$ that are actually levels would lead to a better understanding 
of the $\imax$-successor grammar.

Finally, this idea of representation, with the linear algebra analogy, could certainly be
adapted to some sets of terms built over a maximum, a supremum or other similar operations.
In addition to providing decision procedures, it would permit to define a concept similar
to the basis of a vector space on these spaces on these sets, and then it seems to be 
an interesting direction to explore.

\bibliography{stage.bib}

\begin{thebibliography}{10}

\bibitem{assaf_explicit_subtyping}
Ali Assaf.
\newblock {A calculus of constructions with explicit subtyping}.
\newblock In Hugo Herbelin, Pierre Letouzey, and Matthieu Sozeau, editors, {\em {20th International Conference on Types for Proofs and Programs (TYPES 2014)}}, volume~39 of {\em LIPICS}, Institut Henri Poincar{\'e}, Paris, France, May 2014.
\newblock URL: \url{https://hal.archives-ouvertes.fr/hal-01097401}.

\bibitem{theseAssaf}
Ali Assaf.
\newblock {\em {A framework for defining computational higher-order logics}}.
\newblock Theses, {{\'E}cole polytechnique}, 09 2015.
\newblock URL: \url{https://pastel.archives-ouvertes.fr/tel-01235303}.

\bibitem{dedukti}
Ali Assaf, Guillaume Burel, Rapha{\"e}l Cauderlier, David Delahaye, Gilles Dowek, Catherine Dubois, Fr{\'e}d{\'e}ric Gilbert, Pierre Halmagrand, Olivier Hermant, and Ronan Saillard.
\newblock Dedukti : a logical framework based on the $\lambda\pi$-calculus modulo theory.
\newblock 2016.

\bibitem{assaf_universes}
Ali Assaf, Gilles Dowek, Jean-Pierre Jouannaud, and Jiaxiang Liu.
\newblock {Encoding Proofs in Dedukti: the case of Coq proofs}.
\newblock In {\em {Proceedings Hammers for Type Theories}}, Proc. Higher-Order rewriting Workshop, Coimbra, Portugal, July 2016. {Easy Chair}.
\newblock URL: \url{https://inria.hal.science/hal-01330980}.

\bibitem{BarendregtCube}
Henk Barendregt.
\newblock Introduction to generalized type systems.
\newblock {\em Journal of Functional Programming}, 1(2):125–154, 1991.
\newblock \href {https://doi.org/10.1017/S0956796800020025} {\path{doi:10.1017/S0956796800020025}}.

\bibitem{BarendregtPTS}
Henk Barendregt, S.~Abramsky, D.~Gabbay, T.~Maibaum, and Henk~(Hendrik) Barendregt.
\newblock Lambda calculi with types.
\newblock 10 2000.

\bibitem{geocoq}
Michael Beeson, Pierre Boutry, Gabriel Braun, Charly Gries, and Julien Narboux.
\newblock {GeoCoq}, June 2018.
\newblock URL: \url{https://inria.hal.science/hal-01912024}.

\bibitem{berardiPTS}
Stefano Berardi.
\newblock {\em Type dependence and constructive mathematics}.
\newblock PhD thesis, PhD thesis, Dipartimento di Informatica, Torino, Italy, 1990.

\bibitem{blanqui_universes}
Fr\'{e}d\'{e}ric Blanqui.
\newblock {Encoding Type Universes Without Using Matching Modulo Associativity and Commutativity}.
\newblock In Amy~P. Felty, editor, {\em 7th International Conference on Formal Structures for Computation and Deduction (FSCD 2022)}, volume 228 of {\em Leibniz International Proceedings in Informatics (LIPIcs)}, pages 24:1--24:14, Dagstuhl, Germany, 2022. Schloss Dagstuhl -- Leibniz-Zentrum f{\"u}r Informatik.
\newblock URL: \url{https://drops.dagstuhl.de/opus/volltexte/2022/16305}, \href {https://doi.org/10.4230/LIPIcs.FSCD.2022.24} {\path{doi:10.4230/LIPIcs.FSCD.2022.24}}.

\bibitem{theoryU}
Fr\'{e}d\'{e}ric Blanqui, Gilles Dowek, \'{E}milie Grienenberger, Gabriel Hondet, and Fran\c{c}ois Thir\'{e}.
\newblock {Some Axioms for Mathematics}.
\newblock In Naoki Kobayashi, editor, {\em 6th International Conference on Formal Structures for Computation and Deduction (FSCD 2021)}, volume 195 of {\em Leibniz International Proceedings in Informatics (LIPIcs)}, pages 20:1--20:19, Dagstuhl, Germany, 2021. Schloss Dagstuhl -- Leibniz-Zentrum f{\"u}r Informatik.
\newblock URL: \url{https://drops.dagstuhl.de/opus/volltexte/2021/14258}, \href {https://doi.org/10.4230/LIPIcs.FSCD.2021.20} {\path{doi:10.4230/LIPIcs.FSCD.2021.20}}.

\bibitem{Burel}
Mathieu Boespflug and Guillaume Burel.
\newblock Coqine: Translating the calculus of inductive constructions into the $\lambda\pi$-calculus modulo.
\newblock In {\em in "Second International Workshop on Proof Exchange for Theorem Proving}, 2012.

\bibitem{lambdapimoduniversal}
Mathieu Boespflug, Quentin Carbonneaux, and Olivier Hermant.
\newblock {The $\lambda\Pi$-calculus Modulo as a Universal Proof Language}.
\newblock {\em CEUR Workshop Proceedings}, 878, 06 2012.

\bibitem{PaulinCoquand}
Thierry Coquand and Christine Paulin.
\newblock Inductively defined types.
\newblock In Per Martin-L{\"o}f and Grigori Mints, editors, {\em COLOG-88}, pages 50--66, Berlin, Heidelberg, 1990. Springer Berlin Heidelberg.

\bibitem{Courant}
Judica{\"e}l Courant.
\newblock Explicit universes for the calculus of constructions.
\newblock In Victor~A. Carre{\~{n}}o, C{\'e}sar~A. Mu{\~{n}}oz, and Sofi{\`e}ne Tahar, editors, {\em Theorem Proving in Higher Order Logics}, pages 115--130, Berlin, Heidelberg, 2002. Springer Berlin Heidelberg.

\bibitem{DowekCousineau}
Denis Cousineau and Gilles Dowek.
\newblock Embedding pure type systems in the lambda-pi-calculus modulo.
\newblock pages 102--117, 06 2007.
\newblock \href {https://doi.org/10.1007/978-3-540-73228-0_9} {\path{doi:10.1007/978-3-540-73228-0_9}}.

\bibitem{jouannaudRewrite}
Nachum Dershowitz and Jean{-}Pierre Jouannaud.
\newblock Rewrite systems.
\newblock In Jan van Leeuwen, editor, {\em Handbook of Theoretical Computer Science, Volume {B:} Formal Models and Semantics}, pages 243--320. Elsevier and {MIT} Press, 1990.
\newblock \href {https://doi.org/10.1016/b978-0-444-88074-1.50011-1} {\path{doi:10.1016/b978-0-444-88074-1.50011-1}}.

\bibitem{theseFerey}
Gaspard F{\'{e}}rey.
\newblock {\em Higher-Order Confluence and Universe Embedding in the Logical Framework . (Confluence d’ordre supérieur et encodage d’univers dans le Logical Framework}.
\newblock PhD thesis, {\'{E}}cole normale sup{\'{e}}rieure Paris-Saclay, France, 2021.
\newblock URL: \url{https://lmf.cnrs.fr/downloads/Perso/Ferey-thesis.pdf}.

\bibitem{kontroli}
Michael Färber.
\newblock Safe, fast, concurrent proof checking for the lambda-pi calculus modulo rewriting.
\newblock pages 225--238, 01 2022.
\newblock \href {https://doi.org/10.1145/3497775.3503683} {\path{doi:10.1145/3497775.3503683}}.

\bibitem{SCT}
Guillaume Genestier.
\newblock {SizeChangeTool: A Termination Checker for Rewriting Dependent Types}.
\newblock In Mauricio Ayala-Rinc{\'o}n, Silvia Ghilezan, and Jakob~Grue Simonsen, editors, {\em {HOR 2019 - 10th International Workshop on Higher-Order Rewriting}}, Joint Proceedings of HOR 2019 and IWC 2019, pages 14--19, Dortmund, Germany, June 2019.
\newblock URL: \url{https://hal.science/hal-02442465}.

\bibitem{genestier_universes}
Guillaume Genestier.
\newblock {Encoding Agda Programs Using Rewriting}.
\newblock In Zena~M. Ariola, editor, {\em 5th International Conference on Formal Structures for Computation and Deduction (FSCD 2020)}, volume 167 of {\em Leibniz International Proceedings in Informatics (LIPIcs)}, pages 31:1--31:17, Dagstuhl, Germany, 2020. Schloss Dagstuhl--Leibniz-Zentrum f{\"u}r Informatik.
\newblock URL: \url{https://drops.dagstuhl.de/opus/volltexte/2020/12353}, \href {https://doi.org/10.4230/LIPIcs.FSCD.2020.31} {\path{doi:10.4230/LIPIcs.FSCD.2020.31}}.

\bibitem{sttfageocoq}
Yoan Géran.
\newblock Stt$\forall$ geocoq, 2021.
\newblock URL: \url{https://github.com/Karnaj/sttfa_geocoq_euclid}.

\bibitem{ELF}
Robert Harper, Furio Honsell, and Gordon Plotkin.
\newblock A framework for defining logics.
\newblock {\em J. ACM}, 40(1):143–184, 1 1993.
\newblock \href {https://doi.org/10.1145/138027.138060} {\path{doi:10.1145/138027.138060}}.

\bibitem{HarperPollack}
Robert Harper and Robert Pollack.
\newblock Type checking with universes.
\newblock In {\em 2nd International Joint Conference on Theory and Practice of Software Development}, TAPSOFT '89, page 107–136, NLD, 1991. Elsevier Science Publishers B. V.

\bibitem{lambdapi}
Gabriel Hondet and Fr{\'e}d{\'e}ric Blanqui.
\newblock {The New Rewriting Engine of Dedukti (System Description)}.
\newblock In Zena~M. Ariola, editor, {\em 5th International Conference on Formal Structures for Computation and Deduction (FSCD 2020)}, volume 167 of {\em Leibniz International Proceedings in Informatics (LIPIcs)}, pages 35:1--35:16, Dagstuhl, Germany, 2020. Schloss Dagstuhl--Leibniz-Zentrum f{\"u}r Informatik.
\newblock URL: \url{https://drops.dagstuhl.de/opus/volltexte/2020/12357}, \href {https://doi.org/10.4230/LIPIcs.FSCD.2020.35} {\path{doi:10.4230/LIPIcs.FSCD.2020.35}}.

\bibitem{HuetType}
Gérard Huet.
\newblock Extending the calculus of constructions with type: Type, 1988.
\newblock Unpublished draft.
\newblock URL: \url{https://pauillac.inria.fr/~huet/PUBLIC/typtyp.pdf}.

\bibitem{TTT2}
Martin Korp, Christian Sternagel, Harald Zankl, and Aart Middeldorp.
\newblock Tyrolean termination tool 2.
\newblock In Ralf Treinen, editor, {\em Rewriting Techniques and Applications, 20th International Conference, {RTA} 2009, Bras{\'{\i}}lia, Brazil, June 29 - July 1, 2009, Proceedings}, volume 5595 of {\em Lecture Notes in Computer Science}, pages 295--304. Springer, 2009.
\newblock \href {https://doi.org/10.1007/978-3-642-02348-4\_21} {\path{doi:10.1007/978-3-642-02348-4\_21}}.

\bibitem{luocumul}
Zhaohui Luo.
\newblock Notes on universes in type theory, October, 2012.
\newblock URL: \url{https://www.cs.rhul.ac.uk/home/zhaohui/universes.pdf}.

\bibitem{Paulin}
Christine Paulin-Mohring.
\newblock Inductive definitions in the system coq rules and properties.
\newblock In Marc Bezem and Jan~Friso Groote, editors, {\em Typed Lambda Calculi and Applications}, pages 328--345, Berlin, Heidelberg, 1993. Springer Berlin Heidelberg.

\bibitem{SozeauTabareau}
Matthieu Sozeau and Nicolas Tabareau.
\newblock Universe polymorphism in coq.
\newblock In Gerwin Klein and Ruben Gamboa, editors, {\em Interactive Theorem Proving}, pages 499--514, Cham, 2014. Springer International Publishing.

\bibitem{terese}
Terese.
\newblock {\em Term rewriting systems.}, volume~55 of {\em Cambridge tracts in theoretical computer science}.
\newblock Cambridge University Press, 2003.

\bibitem{Thire}
Fran{\c{c}}ois Thir{\'{e}}.
\newblock Sharing a library between proof assistants: Reaching out to the {HOL} family.
\newblock In Fr{\'{e}}d{\'{e}}ric Blanqui and Giselle Reis, editors, {\em Proceedings of the 13th International Workshop on Logical Frameworks and Meta-Languages: Theory and Practice, LFMTP@FSCD 2018, Oxford, UK, 7th July 2018}, volume 274 of {\em {EPTCS}}, pages 57--71, 2018.
\newblock \href {https://doi.org/10.4204/EPTCS.274.5} {\path{doi:10.4204/EPTCS.274.5}}.

\bibitem{theseThire}
Fran{\c{c}}ois Thir{\'{e}}.
\newblock {\em Interoperability between proof systems using the logical framework Dedukti. (Interop{\'{e}}rabilit{\'{e}} entre syst{\`{e}}mes de preuve en utilisant le cadre logique Dedukti)}.
\newblock PhD thesis, {\'{E}}cole normale sup{\'{e}}rieure Paris-Saclay, Cachan, France, 2020.
\newblock URL: \url{https://tel.archives-ouvertes.fr/tel-03224039}.

\bibitem{voevodsky}
Vladimir Voevodsky.
\newblock A universe polymorphic type system, October 22, 2014.
\newblock An unfinished unreleased manuscript.
\newblock URL: \url{https://www.math.ias.edu/Voevodsky/files/files-annotated/Dropbox/Unfinished_papers/Type_systems/UPTS_current/Universe_polymorphic_type_sytem.pdf}.

\bibitem{CSI}
Harald Zankl, Bertram Felgenhauer, and Aart Middeldorp.
\newblock Csi: a confluence tool.
\newblock volume 6803, pages 499--505, 07 2011.
\newblock \href {https://doi.org/10.1007/978-3-642-22438-6_38} {\path{doi:10.1007/978-3-642-22438-6_38}}.

\end{thebibliography}

\appendix


\section{Proofs of \cref{sec-level-transformation}}

\propmaxright*

\begin{proof}
	Let $\sigma$ be a valuation, $t = \prodrule[u, \max[v, w]]$, $t_1 = \prodrule[u, v]$
	and $t_2 = \prodrule[u, w]$.
	\begin{itemize}
		\item If $\valuation{v}{\sigma} = \valuation{w}{\sigma} = 0$, then
		$\valuation{\max[t_1, t_2]}{\sigma} = 0
		= \valuation{t}{\sigma}$.
		\item If $\valuation{v}{\sigma} \neq 0$ and $\valuation{w}{\sigma} = 0$,
		then
		$
		\valuation{\max[t_1, t_2]}{\sigma} = 
		\max[\valuation{u}{\sigma}, \valuation{v}{\sigma}] =
		\valuation{t}{\sigma}.
		$
		\item If $\valuation{v}{\sigma} = 0$ and $\valuation{w}{\sigma} \neq 0$,
		then
		$
		\valuation{\max[t_1, t_2]}{\sigma} = 
		\max[\valuation{u}{\sigma}, \valuation{w}{\sigma}] =
		\valuation{t}{\sigma}.
		$
		\item Else,
		$\valuation{\max[t_1, t_2]}{\sigma} = 
		\max[\valuation{u}{\sigma}, \valuation{v}{\sigma}, \valuation{w}{\sigma}] =
		\valuation{t}{\sigma}$.
	\end{itemize}
\end{proof}

\propmaxleft*

\begin{proof}
	Let $\sigma$ be a valuation. 
	\begin{itemize}
		\item If $\valuation{w}{\sigma} = 0$, then
		$
		\valuation{\max[\prodrule[u, w], \prodrule[u, w]]}{\sigma} 
		= 0
		=\valuation{\prodrule[\max[u, v], w]}{\sigma}.
		$
		\item Else,
		$
		\valuation{\max[\prodrule[u, w], \prodrule[v, w]]}{\sigma} 
		= \max[\valuation{u}{\sigma}, \valuation{v}{\sigma}, \valuation{w}{\sigma}]
		= \valuation{\prodrule[\max[u, v], w]}{\sigma}.
		$	
	\end{itemize}
\end{proof}

\proprright*
\begin{proof}
	Let $\sigma$ be a valuation. 
	\begin{itemize}
		\item If $\valuation{w}{\sigma} = 0$, then $
		\valuation{\max[\prodrule[u, w], \prodrule[u, w]]}{\sigma} 
		= 0
		= \valuation{\prodrule[u, \prodrule[v, w]]}{\sigma}.
		$
		\item Else, $
		\valuation{\max[\prodrule[u, w], \prodrule[v, w]]}{\sigma} 
		= \max[\valuation{u}{\sigma}, \valuation{v}{\sigma}, \valuation{w}{\sigma}] 
		= \valuation{\prodrule[u, \prodrule[v, w]]}{\sigma}.
		$
	\end{itemize}
\end{proof}

\propdistrplus*
\begin{proof}
	Let $\sigma$ be a valuation. 
	\begin{itemize}
		\item If $\valuation{w}{\sigma} = 0$, then
		$\valuation{s(\prodrule[v, w])}{\sigma} = s(0) =
		\valuation{\max[s(w), \prodrule[s(v), w]]}{\sigma}$.
		\item Else $\valuation{s(\prodrule[v, w])}{\sigma} = 
		s(\max[\valuation{v}{\sigma}, \valuation{w}{\sigma}]) =
		\valuation{\max[s(w), \prodrule[s(v), w]]}{\sigma}$.
	\end{itemize}
\end{proof}

\propassym*
\begin{proof}
	The two cases are very similar. We show the result for the first one.
	Let $\sigma$ be a valuation. If forall $1 \leqslant i \leqslant n$,
	$\sigma(x_i) \neq 0$, then
	\begin{gather*}
		\valuation{t}{\sigma} = \max[\sigma(x_n), \ldots, \sigma(x_1), k + \sigma(y)] \\
		\forall 1 < i \leqslant n, 
		\valuation{\A[\set{x_n, \ldots, x_{i + 1}}, x_i, 0]}{\sigma} = \sigma(x_i)\\
		\valuation{\A[\set{x_n, \ldots, x_1}, y, k]}{\sigma} = k + \sigma(y)
	\end{gather*}
	and else, we take the largest $1 \leqslant i \leqslant n$ such that $\sigma(x_i) = 0$, then
	\begin{align*}
		&\valuation{t}{\sigma} = \max[\sigma(x_n), \ldots, \sigma(x_{i + 1})] &&
		\forall 1 < j \leqslant i, 
		\valuation{\A[\set{x_n, \ldots, x_{j + 1}}, x_j, 0]}{\sigma} = 0\\
		&\valuation{\A[\set{x_n, \ldots, x_1}, y, k]}{\sigma} = 0 &&
		\forall i < j \leqslant n,
		\valuation{\A[\set{x_n, \ldots, x_{j + 1}}, x_j, 0]}{\sigma} = \sigma(x_j)
	\end{align*}
	hence the equality.
\end{proof}

\propavariable*
\begin{proof}
	Let $\sigma$ be a valuation, $t = \A[E, x, S]$, $u = \A[E \cup \set{x}, x, S]$ and $v = \B[E, S]$.
	\begin{itemize}
		\item If there exists $y \in E$ such that $\sigma(y) = 0$, then
		$\valuation{t}{\sigma} = \valuation{u}{\sigma} = \valuation{v}{\sigma} = 0$.
		\item Else, if $\sigma(x) = 0$, then
		$\valuation{t}{\sigma} = S$,
		$\valuation{u}{\sigma} = 0$ and
		$\valuation{v}{\sigma} = S$.
		\item Else, $\sigma(x) \neq 0$, and then
		$\valuation{t}{\sigma} = \sigma(x) + S$,
		$\valuation{u}{\sigma} = \sigma(x) + S$ and 
		$\valuation{v}{\sigma} = S$.
	\end{itemize}
	Hence the result.
\end{proof}

\section{Proofs of \cref{sec-uniqueness}}

\thcomparisonsublevel*
\begin{proof}
	With $\sigma$ such that $\sigma(x) = K + 1$ and $\sigma(y) = 1$ if $y \neq x$,
	we show the first case. Indeed, $\valuation{\A[E, x, S]}{\sigma} = K + 1 + S > K = \valuation{\B[F, K]}{\sigma}$
	hence $\A[E, x, S] \not\leqTerm \B[F, K]$.
	The cases 2, 3 and 4 corresponds to \cref{prop-compare_aa,prop-compare_ab,prop-compare_bb}
	proved below.
\end{proof}

	\begin{proposition} \label{prop-compare_aa}
		Let $E, F \subset \X$, $x \in E$, $y \in F$ and $S, K \in \N$. Then
		\[
		\A[E, x, S] \leqTerm \A[F, y, K] \iff F \subset E \land x = y \land S \leqslant K.
		\]
	\end{proposition}
	\begin{proof}
		We note $t_1 = \A[E, x, S]$ and $t_2 = \A[F, y, K]$. Let us suppose $F \subset E$, $x = y$ and $S \leqslant K$. Let $\sigma$ be a substitution. 
		\begin{itemize}
			\item If there exists $y \in F$ such that $\sigma(y) = 0$, then 
			$\valuation{t_2}{\sigma} = 0$ and since $F \subset E$, 
			$\valuation{t_1}{\sigma} = 0$.
			\item Else, $\valuation{t_1}{\sigma} \leqslant \sigma(x) + S \leqslant \sigma(x) + K 
			= \valuation{t_2}{\sigma}$.
		\end{itemize}
		In both cases, $\valuation{t_1}{\sigma} \leqslant \valuation{t_2}{\sigma}$ 
		hence $t_1 \leqTerm t_2$.
		
		Now, we show the other implication by contraposition.
		\begin{itemize}
			\item If there exists $z \in F$ such that $z \not \in E$, we take $\sigma$ such that 
			$\sigma(z) = 0$ and forall $j \neq z$, $\sigma(j) = 1$. We note that $z \neq x$ 
			(since $z\not\in E$ and $x \in E$) hence $\sigma(x) = 1$. Then,
			$
			\valuation{t_1}{\sigma} = S + 1 > 0 = \valuation{t_2}{\sigma}.
			$
			\item If $x \neq y$ we take $\sigma$ such that $\sigma(x) = K + 2$, $\sigma(y) = 1$ and 
			forall $z \neq x$ and $z \neq y$, $\sigma(z) = 1$. Then,
			$
			\valuation{t_1}{\sigma} = K + S + 2 > K + 1 = \valuation{t_2}{\sigma}.
			$
			\item If $S > K$ we take $\sigma$ such that forall $z$, $\sigma(z) = 1$. Then,
			$
			\valuation{t_1}{\sigma} = S + 1 > K + 1 = \valuation{t_2}{\sigma}.
			$
		\end{itemize}
	\end{proof}
	
	\begin{proposition} \label{prop-compare_bb}
		Let $E, F \subset \X$ and $S, K \in \N$. Then
		\[
		\B[E, S] \leqTerm \B[F, K] \iff F \subset E \land S \leqslant K.
		\]
	\end{proposition}
	\begin{proof}
		We note $t_1 = \B[E, S]$ and $t_2 = \B[F, K]$. Let us suppose $F \subset E$ and $S \leqslant K$.
		Let $\sigma$ be a substitution. 
		\begin{itemize}
			\item If there exists $y \in F$ such that $\sigma(y) = 0$, then $\valuation{t_2}{\sigma} = 0$
			and since $F \subset E$, $\valuation{t_1}{\sigma} = 0$.
			\item Else, $\valuation{t_1}{\sigma} \leqslant K \leqslant S = \valuation{t_2}{\sigma}$.
		\end{itemize}
		In both cases, $\valuation{t_1}{\sigma} \leqslant \valuation{t_2}{\sigma}$ 
		hence $t_1 \leqTerm t_2$.
		
		Now, we show the other implication by contraposition.
		\begin{itemize}
			\item If there exists $y \in F$ such that $y \not \in E$, we take $\sigma$ such that 
			$\sigma(y) = 0$ and forall $z \neq y$, $\sigma(z) = 1$. Then,
			$
			\valuation{t_1}{\sigma} = S > 0 = \valuation{t_2}{\sigma}.
			$
			\item If $S < K$ we take $\sigma$ such that forall $y$, $\sigma(y) = 1$. Then,
			$
			\valuation{t_2}{\sigma} = K > S = \valuation{t_1}{\sigma}.
			$
		\end{itemize}
	\end{proof}

	\begin{proposition} \label{prop-compare_ab}
		Let $E, F \subset \X$, $x \in E$ and $K, S \in \N$. Then
		\[
		\B[E, S] \leqTerm \A[F, x, K] \iff (F \subset E \land S \leqslant K + 1).
		\]
	\end{proposition}
	\begin{proof}
		We note $t_1 = \B[E, S]$ and $t_2 = \A[F, x, K]$. Let us suppose $F \subset E$ and $S \leqslant K + 1$.
		Let $\sigma$ be a substitution. 
		\begin{itemize}
			\item If there exists $y \in F$ such that $\sigma(y) = 0$, then $\valuation{t_2}{\sigma} = 0$
			and since $F \subset E$, $\valuation{t_1}{\sigma} = 0$.
			\item Else, $\sigma(x) \geq 1$ (because $x \in F$) and then 
			$\valuation{t_2}{\sigma} = \sigma(x) + K
			\geq 1 + K \geq S 
			\geq \valuation{t_1}{\sigma}$.
		\end{itemize}
		In both cases, $\valuation{t_1}{\sigma} \leqslant \valuation{t_2}{\sigma}$ 
		hence $t_1 \leqTerm t_2$.
		
		Now, we show the other implication by contraposition. First, we note that $S > 0$.
		\begin{itemize}
			\item If there exists $y \in F$ such that $y \not \in E$, we take $\sigma$ such that 
			$\sigma(y) = 0$ and forall $z \neq y$, $\sigma(z) = 1$. Then,
			$
			\valuation{t_1}{\sigma} = K > 0 = \valuation{t_2}{\sigma}.
			$
			\item If $S > K + 1$ we take $\sigma$ such that forall $y$, $\sigma(y) = 1$. Then,
			$
			\valuation{t_1}{\sigma} = S > K + 1 = \valuation{t_1}{\sigma}.
			$
		\end{itemize}
	\end{proof}

\end{document}

DÉPLACER QUAND JE VAIS PARLER D'INDÉPENDENCE À LA FIN

Intuitively, this last condition corresponds to some independence of the subterms. 

Indeed, if we consider the predicative case, this condition is ensured by the fact that
if $k + x$ is an element of $\set{u_1, \ldots, u_n}$, then the only way to have an equality is to have $k + x$ in $\set{v_1, \ldots, v_m}$ because $k + x$ cannot 
In the predicative case, if $k + x$ is one

\end{document}